\documentclass{article}
\usepackage{ijcai22}

\usepackage{times}
\usepackage{soul}
\usepackage{url}
\usepackage[hidelinks]{hyperref}
\hypersetup{colorlinks=true,allcolors=teal}
\usepackage[utf8]{inputenc}
\usepackage[small]{caption}
\usepackage{graphicx}
\usepackage{amsmath,amsthm,amsfonts,amssymb,mathtools}
\usepackage{booktabs}

\usepackage{multirow}
\usepackage[noend]{algpseudocode}
\usepackage{natbib}
\usepackage{comment}
\usepackage[capitalize]{cleveref}
\usepackage{bm}
\usepackage{bbm}
\usepackage{xcolor}
\usepackage{todonotes}
\usepackage{dsfont} \newcommand{\bone}{\mathds{1}}
\usepackage[ruled,vlined,linesnumbered]{algorithm2e}

\DeclareMathOperator*{\argmax}{argmax}

\usepackage{thm-restate}
\theoremstyle{plain}
\newtheorem{theorem}{Theorem}
\newtheorem{lemma}[theorem]{Lemma}

\newtheorem{corollary}[theorem]{Corollary}
\newtheorem{proposition}[theorem]{Proposition}
\theoremstyle{definition}

\title{Envy-Free and Pareto-Optimal Allocations for \\ Agents with Asymmetric Random Valuations}

\author{
Yushi Bai$^1$
\And
Paul G\"olz$^2$
\affiliations
$^1$Institute for Interdisciplinary Information Sciences, Tsinghua University, China.\\
$^2$Computer Science Department, Carnegie Mellon University, Pittsburgh, PA, USA.\\
\emails
bys18@mails.tsinghua.edu.cn, pgoelz@cs.cmu.edu
}

\newcommand{\new}[1]{#1}
\newif\iffullversion
\fullversiontrue 
\iffullversion
\newcommand{\full}[2]{#1}
\else
\newcommand{\full}[2]{#2}
\fi

\everypar{\looseness=-1}

\begin{document}

\maketitle

\begin{abstract}
We study the problem of allocating $m$ indivisible items to $n$ agents with additive utilities.
It is desirable for the allocation to be both fair and efficient, which we formalize through the notions of \emph{envy-freeness} and \emph{Pareto-optimality}. 
While envy-free and Pareto-optimal allocations may not exist for arbitrary utility profiles, previous work has shown that such allocations exist with high probability assuming that all agents' values for all items are independently drawn from a common distribution.
In this paper, we consider a generalization of this model where each agent's utilities are drawn independently from a distribution \emph{specific to the agent}.
We show that envy-free and Pareto-optimal allocations are likely to exist in this asymmetric model when $m=\Omega\left(n\, \log n\right)$, which is tight up to a log log gap that also remains open in the symmetric subsetting.
Furthermore, these guarantees can be achieved by a polynomial-time algorithm.
\end{abstract}

\section{Introduction}
\label{sec:intro}
Imagine that the neighborhood children go trick-or-treating and return successfully, with a large heap of candy between them.
They then try to divide the candy amongst themselves, but quickly reach the verge of a fight:
Each has their own conception of which sweets are most desirable, and, whenever a child suggests a way of splitting the candy, another child feels unfairly disadvantaged.
As a (mathematically inclined) adult in the room, you may wonder:
Which allocation of candies should you suggest to keep the peace? And, is it even possible to find such a fair distribution?

In this paper, we study the classic problem of fairly dividing $m$ items among $n$ agents~\citep{BCM16}, as exemplified by the scenario above.
We assume that the items we seek to divide are goods (i.e., receiving an additional piece of candy never makes a child less happy), that items are indivisible (candy cannot be split or shared), and that the agents have \emph{additive} valuations (roughly: a child's value for a piece of candy does not depend on which other candies they receive).

We will understand an allocation to be fair if it satisfies two axioms: \emph{envy-freeness} (EF) and \emph{Pareto-optimality} (PO). 
First, fair allocations should be envy-free, which means that no agent should strictly prefer another agent's bundle to their own.
After all, if an allocation violates envy-freeness, the former agent has good reason to contest it as unfair.
Second, fair allocations should be Pareto-optimal, i.e., there should be no reallocation of items making some agent strictly better off and no agent worse off.
Not only does this axiom rule out allocations whose wastefulness is unappealing; it is also arguably necessary to preserve envy-freeness:
Indeed, if a chosen allocation is envy-free but not Pareto-optimal, rational agents can be expected to trade items after the fact, which might lead to a final allocation that is not envy-free after all.
Unfortunately, even envy-freeness alone is not always attainable. For instance, if two agents like a single item, the agent who does not receive it will always envy the agent who does.

Motivated by the fact that worst-case allocation problems may not have fair allocations, a line of research in fair division studies asymptotic conditions for the existence of such allocations, under the assumption that the agents' utilities are random rather than adversarially chosen~\cite[e.g.][]{DGK+14,MS19}.
Specifically, these papers assume that all agents' utilities for all items are independently drawn from a common distribution $\mathcal{D}$, a model which we call the \emph{symmetric model}.
Among the algorithms shown to satisfy envy-freeness in this setting, only one is also Pareto-optimal: the (utilitarian) \emph{welfare-maximizing algorithm}, which assigns each item to the agent who values it the most.
This algorithm is Pareto-optimal, and it is also envy-free with high probability as the number of items $m$ grows in $\Omega\left(n \, \log n\right)$.\footnote{In fact, \citet{DGK+14} prove this result for a somewhat more general model than the one presented above (certain correlatedness between distributions is also allowed), but their model assumes the key symmetry between agents that we discuss below.}
Since envy-free allocations may exist with only vanishing probability for $m \in \Theta\left(n \, \log n / \log \log n\right)$ in the symmetric model~\citep{MS19}, the above result characterizes almost tightly when envy-free and Pareto-optimal allocations exist in this model.

Zooming out, however, this positive result is unsatisfying in that, outside of this specific random model, the welfare-maximizing algorithm can hardly be called ``fair'':
For example, if an agent~A tends to have higher utility for most items than agent~B, the welfare-maximizing algorithm will allocate most items to agent~A, which can cause large envy for agent~B.
In short, the welfare-maximizing algorithm leads to fair allocations only because the model assumes each agent to be equally likely to have the largest utility, an assumption that limits the lessons that can be drawn from this model.

Motivated by these limitations of prior work, this paper investigates the existence of fair allocations in a generalization of the symmetric model, which we refer to as the \emph{asymmetric model}.
In this model, each agent $i$ is associated with their own distribution $\mathcal{D}_i$, from which their utility for all items is independently drawn.
\begin{figure}
    \centering
    \includegraphics[width=0.9\linewidth]{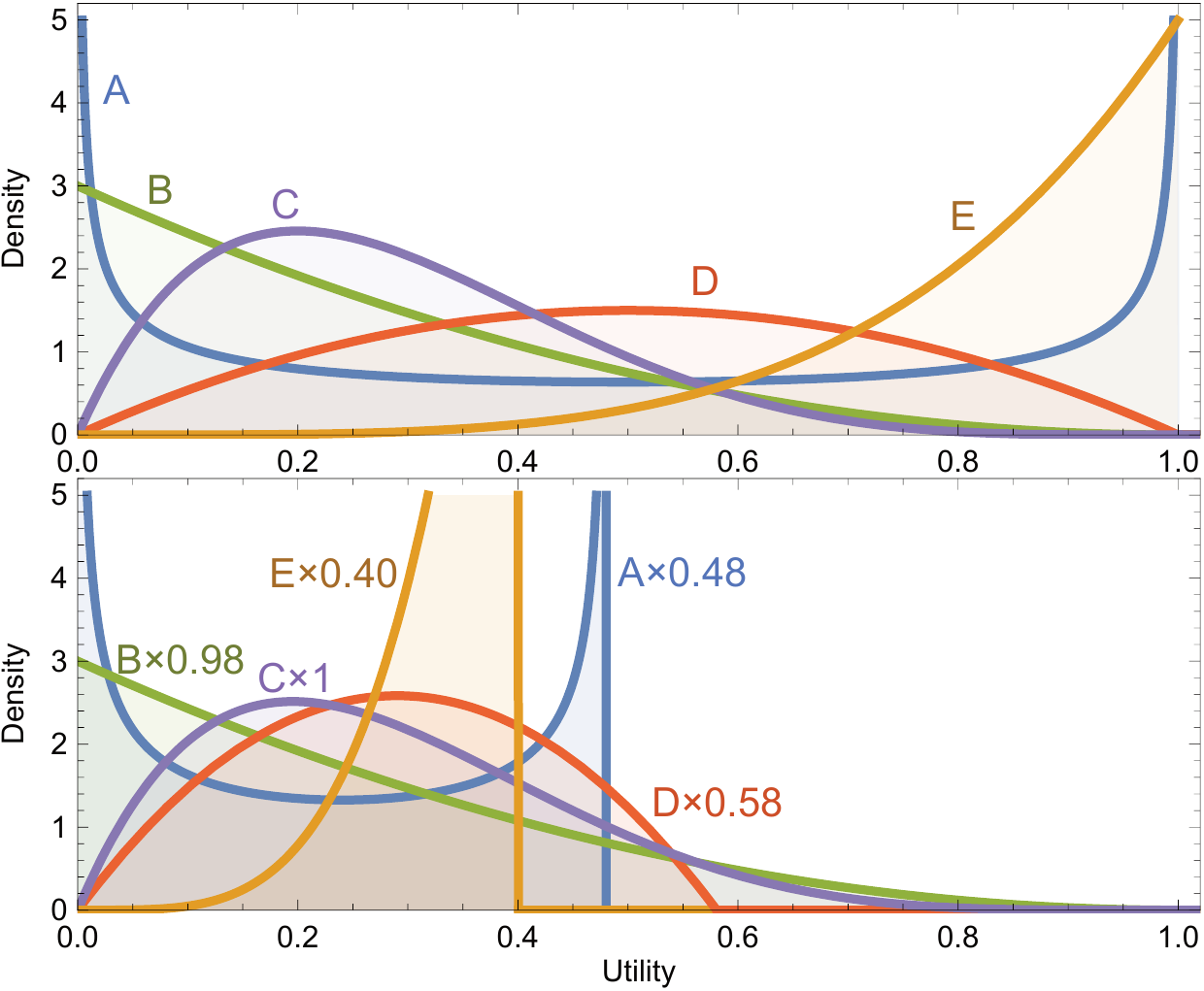}
    \caption{The top panel shows probability density functions of five agents' utility distributions. The bottom panel shows densities after scaling distributions by the given multipliers. When drawing an independent sample from each scaled distribution, each sample is the largest with probability $1/5$.}
    \label{fig:scaling}
\end{figure}
Within this model, we aim to answer the question: \emph{When do envy-free and Pareto-optimal allocations exist for agents with asymmetric valuations?}

\subsection{Our Techniques and Results}
\label{sec:results}
\begin{figure}
\centering
\includegraphics[width=.95\linewidth]{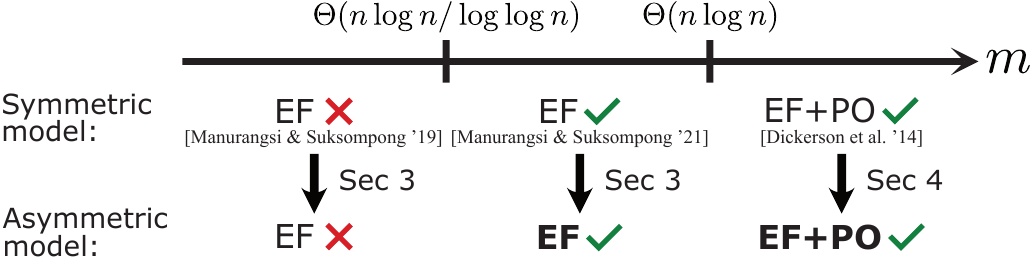}
\caption{Existing and new results on when EF and EF+PO allocations are guaranteed to exist in both models. Bold results are new.}
\label{fig:result}
\end{figure}
In \cref{sec:symm}, we study which results in the symmetric model generalize to the asymmetric model.
In particular, we apply an analysis by \citet{MS20} to the asymmetric model in a black-box manner to prove envy-free allocations exist when $m \in \Omega\left(n \, \log n / \log \log n\right)$, which is tight with existing impossibility results on envy-freeness. However, this approach does not preserve Pareto-optimality.

Using a new approach, we prove in \cref{sec:existence} that generalizing the random model from symmetric to asymmetric agents does not substantially decrease the frequency of envy-free and Pareto-optimal allocations.
The key idea is to find a multiplier $\beta_i > 0$ for each agent $i$ such that, when drawing an independent sample $u_i$ from each utility distribution $\mathcal{D}_i$, each agent $i$ has an equal probability of $\beta_i u_i$ being larger than the $\beta_{j} u_{j}$ of all other agents $j \neq i$, which we call the agent's \emph{resulting probability} from these multipliers.
\cref{fig:scaling} illustrates how five utility distributions can be rescaled in this way.
If all resulting probabilities of a set of multipliers equal $1/n$, we call these multipliers \emph{equalizing}.

A set of equalizing multipliers defines what we call its \emph{multiplier allocation}, which assigns each item to the agent $i$ whose utility weighted by $\beta_i$ is the largest.
Put differently, the multiplier allocation simulates the welfare-maximizing algorithm in an instance in which each agent $i$'s distribution is scaled by $\beta_i$.
Just like the welfare-maximizing allocations, the multiplier allocations are Pareto-optimal by construction, and the similarity between both allocation types allows us to apply proof techniques developed for the welfare-maximizing algorithm and the symmetric setting to show envy-freeness.

The core of our paper is a proof that equalizing multipliers always exist, which we show using Sperner's lemma.
Since an algorithm based on this direct proof would have exponential running time, we design a polynomial-time algorithm for computing approximately equalizing multipliers, i.e., multipliers whose resulting probabilities lie within $[1/n - \delta, 1/n + \delta]$ for a $\delta > 0$ given in the input.

Having established the existence of equalizing multipliers, we go on to show that the multiplier allocation is envy-free with high probability.
To obtain this result, we demonstrate a constant-size gap between each agent's expected utility for an item conditioned on them receiving the item and the agent's expected utility conditioned on another agent receiving the item, and then use a variant of the argument of \citet{DGK+14} to show that multiplier allocations are envy-free with high probability when $m \in \Omega\left(n \, \log n\right)$.
This guarantee extends to the case where we allocate based on multipliers that are sufficiently close to equalizing, which means that our polynomial-time \emph{approximate multiplier algorithm} is Pareto-optimal and envy-free with high probability.

In \cref{sec:experiments}, we empirically evaluate how many items are needed to guarantee envy-free and Pareto-optimal allocations for a fixed collection of agents.
We find that the approximate multiplier algorithm needs relatively large numbers of items to ensure envy-freeness; that the round robin algorithm violates Pareto-optimality in almost all instances; and that the Maximum Nash Welfare (MNW) algorithm achieves both axioms already for few items but that its running time limits its applicability.
For larger numbers of items, the approximate multiplier algorithm satisfies both axioms and excels by virtue of its running time.

\subsection{Related work}
The question of when fair allocations exist for random utilities was first raised by \citet{DGK+14}, whose main result we have already discussed.
Our paper also builds on work by \citet{MS19,MS20}, who prove the lower bound on the existence of envy-free allocations mentioned in the introduction and that the classic round robin algorithm produces envy-free allocations in the symmetric model for slightly lower $m$ than the welfare-maximizing algorithm.
A bit further afield, \citet{Suksompong16} and \citet{AMN+17} study the existence of proportional and maximin-share allocations (two relaxations of envy-freeness) in the symmetric model, and \citet{MS17a} study envy-freeness when items are allocated to groups rather than to individuals.
None of these papers consider Pareto-optimality, perhaps because fair division yields few tools for simultaneously guaranteeing envy-freeness and Pareto-optimality.

The asymmetric model we investigate has been previously used, for example, by \citet{KPW16} to study the existence of maximin-share allocations.
While part of their proof applies the results by \citeauthor{DGK+14} to construct envy-free allocations in the asymmetric model, as do we, their allocation algorithm is not Pareto-optimal (see \cref{sec:symefpo}).
\citet{FGH+19} also consider maximin-share allocations in the asymmetric model, for agents with weighted entitlements.
\citet{ZP20} study allocation problems in the asymmetric model, when items arrive online.
While they do consider and achieve Pareto-optimality, they only obtain approximate notions of envy-freeness.
\new{Finally, \citet{smoothed} study the existence of envy-free allocations and of both proportional and Pareto-optimal allocations in an expressive utility model based on smoothed analysis.}

\section{Preliminaries}
\label{sec:model}
\paragraph{General Definitions.}
We consider a set $M$ of $m$ indivisible items being allocated to a group $N=\{1, \dots, n\}$ of $n$ agents. Each agent $i\in N$ has a \emph{utility} $u_i(\alpha)\geq0$ for each item $\alpha\in M$, indicating their degree of preference for the item. The collection of agent--item utilities make up a \emph{utility profile}. An \emph{allocation} $\mathcal{A}=\{A_i\}_{i\in N}$ is a partition of the items into $n$ \emph{bundles}: $M = A_1 \cup \cdots \cup A_n$, where agent $i$ gets the items in bundle $A_i$. Under our assumption that the agents' utilities are \emph{additive}, agent $i$'s \emph{utility} for a subset of items $A\subseteq M$ is $u_i\left(A\right)=\sum_{\alpha\in A}u_i\left(\alpha\right)$. 

An allocation $\mathcal{A}=\{A_i\}_{i\in N}$ is said to be \emph{envy-free} (EF) if $u_i(A_i)\geq u_i(A_j)$ for all $i,j \in N$, i.e., if each agent weakly prefers their own bundle to any other agent's bundle. 
We say that an allocation $\mathcal{A}=\{A_i\}_{i\in N}$ is \emph{Pareto dominated} by another allocation $\mathcal{A'}=\{A'_i\}_{i\in N}$ if $u_i(A_i)\leq u_i(A'_i)$ for all $i\in N$, with at least one inequality holding strictly. 
An allocation is \emph{Pareto-optimal} (PO) if it is not Pareto dominated by any other allocation.
An allocation is called \emph{fractionally Pareto-optimal} (fPO) if it is not even Pareto dominated by any ``fractional'' allocation of items.
For our purposes, it suffices to note that an allocation is fPO iff there exist multipliers $\{\beta_i > 0\}_{i \in N}$ such that each item $\alpha$ is allocated to an agent $i$ with maximal $\beta_i u_i(\alpha)$~\citep{Negishi60}.

\paragraph{Asymmetric Model.}
In our asymmetric model, each agent $i$ is associated with a \emph{utility distribution} $\mathcal{D}_i$, a nonatomic probability distribution over $[0, 1]$.
The model assumes that the utilities $u_i(\alpha)$ for all $\alpha\in M$ are independently drawn from $\mathcal{D}_i$. 
For simplicity, we just write $u_i$ as a random variable for $u_i(\alpha)$ if we are not talking about a specific item $\alpha$, where $u_i\sim \mathcal{D}_i$.
Let $f_i$ and $F_i$ denote the probability density function (PDF) and cumulative distribution function (CDF) of $\mathcal{D}_i$.
For our main result, we make the following assumptions on utility distributions:
\emph{(a) Interval support}: The support of each $\mathcal{D}_i$ is an interval $[a_i, b_i]$ for $0 \leq a_i < b_i \leq 1$. 
\emph{(b) $\left(p, q\right)$-PDF-boundedness}: For constants $0<p<q$, the density of each $\mathcal{D}_i$ is bounded between $p$ and $q$ within its support. These two assumptions are weaker than those by \citet{MS20}, who additionally require all distributions to have support $[0,1]$.
A random event occurs \emph{with high probability} if the event's probability converges to 1 as $n \to \infty$.

\section{Takeaways From the Symmetric Model}
\label{sec:symm}
We quickly review results obtained in the symmetric model, and to which degree they carry over to the asymmetric model.

\paragraph{Non-Existence of EF Allocations:}
Since the symmetric model is a special case of the asymmetric model\,---\,in which all $\mathcal{D}_i$ are equal\,---\,this negative result immediately applies:
\begin{proposition}[\citealt{MS19}]\footnote{Here, we present a special case; the original result holds for different choices of distribution and leaves some flexibility in $m$.}
There exists $c > 0$ such that, if $m = (\lfloor c \, \log n / \log \log n \rfloor + 1/2) \, n$ and all utility distributions are uniform on $[0,1]$, then, with high probability, no envy-free allocation exists.
\end{proposition}

\paragraph{Existence of EF Allocations:}
In the symmetric model, \citet{MS20} give an allocation algorithm, round robin, that satisfies EF with high probability.
An interesting property of this algorithm is that an agent's allocation given a utility profile depends not on the \emph{cardinal} information of the agents' utilities, but only on their \emph{ordinal} preference order over items.
Using this property, we prove in \full{\cref{app:rr}}{the full version} that their result generalizes to the asymmetric model since, in a nutshell, an agent $i$'s envy of the other agents is indistinguishable between the asymmetric model and a symmetric model with common distribution $\mathcal{D}_i$.
\begin{restatable}{proposition}{proprr}
\label{prop:rr}
When distributions have interval support and are $(p, q)$-PDF-bounded, if $m \in \Omega\left(n \, \log n / \log \log n\right)$, an envy-free allocation exists with high probability.
\end{restatable}
To our knowledge, we are the first to observe that the analysis by \citeauthor{MS20} generalizes in this way, which improves on the previously best known upper bound of $m \in \Omega\left(n \, \log n\right)$ in the asymmetric model due to \citet{KPW16}.
\paragraph{Existing Approaches Do Not Provide EF+PO:}
\label{sec:symefpo}
Generalizing the existence result for EF and PO allocations by \citet{DGK+14} to the asymmetric model is more challenging than the round robin result above, since cardinal information is crucial for the PO property.
In \full{\cref{app:maxpercentile}}{the full version of the paper}, we illustrate this point by considering how \citet{KPW16} apply the theorem of \citeauthor{DGK+14} to prove the existence of EF allocations in the asymmetric model; namely, they assign each item to the agent for whom the item is in the highest percentile of their utility distribution.
On an example, we show that this approach fundamentally violates PO, and that assigning items based on multipliers is the most natural way to guarantee PO.
\new{\full{In \cref{app:normalize},}{In the full version,} we also give an example showing that normalizing each agent's values to add up to one\,---\,perhaps the most obvious way to obtain multipliers\,---\,is not sufficient to provide EF.}

\section{Existence of EF+PO Allocations}
\label{sec:existence}

We now prove our main theorem:
\begin{theorem}
\label{thm:existence}
Suppose that all utility distributions have interval support and are $(p,q)$-PDF-bounded for some $p,q$.
If  $m \in \Omega\left(n \, \log n\right)$ as $n \to \infty$,\footnote{Alternatively, we may assume $n \in O\left(m / \log m\right)$ as $m \to \infty$ to avoid the assumption that $n \to \infty$, as do \citet{DGK+14}.} then, with high probability, an envy-free and (fractionally) Pareto-optimal allocation exists and can be found in polynomial time.
\end{theorem}

In \cref{sec:existence_multipliers}, we prove that we can always find multipliers $\{\beta_i\}_{i \in N}$ that equalize each agent's probability of receiving a random-utility item in the multiplier allocation (which allocates item $\alpha$ to the agent with maximal $\beta_i u_i(\alpha)$ and is trivially fPO).
We also discuss how to efficiently find multipliers leading to approximately equalizing probabilities.
Next, in \cref{sec:existence_gap}, we show that an agent's expected utility for an item allocated to themselves is larger by a constant than their expected utility for an item allocated to another agent.
In \cref{sec:puttingtogetheref}, we combine these properties to prove envy-freeness.

\subsection{Existence of Equalizing Multipliers}
\label{sec:existence_multipliers}
For a set of multipliers $\vec{\beta} \in \mathbb{R}^n_{>0}$ and an agent $i$, we denote $i$'s resulting probability by
\begin{align}
    p_i(\vec{\beta}) \coloneqq{}& \mathbb{P}[\beta_i u_i = \textstyle{\max_{j\in N}\beta_j u_j}]& \notag \\
    ={}& \int_0^1 f_i(u) \, \textstyle{\prod_{j \in N \setminus \{i\}}}  F_j(\beta_i / \beta_j \, u) \, du. \label{eq:piint}
\end{align}

\subsubsection{Existence Proof Using Sperner's Lemma}
The existence of equalizing multipliers can be established quite easily using Sperner's lemma:
\begin{theorem}
\label{thm:equalizing}
For any set of utility distributions, there exists a set of equalizing multipliers.
\end{theorem}
\begin{proof}[Proof sketch]
Since scaling all multipliers by the same factor does not change the resulting probabilities, we may restrict our focus to multipliers within the $(n-1)$-dimensional simplex $S = \{ \vec{\beta} \in \mathbb{R}_{\geq0}^n \mid \sum_{i \in N} \beta_i = 1\}$.
We define a coloring function $f: S \rightarrow N$, which maps each set of multipliers in the simplex to an agent with maximum resulting probability.
Clearly, points on a face $\beta_i = 0$ are not colored with color belonging to agent $i$ since some other agent has a positive multiplier and must thus have a greater scaled utility than $i$.

Now, consider a simplicization of $S$, i.e., a partition of $S$ into small simplices meeting face to face (generalizing the notion of a triangulation in the 2-D simplex).
Sperner's lemma shows the existence of a small simplex that is panchromatic, i.e., whose $n$ vertices are each colored with a different agent.
This small simplex constitutes a neighborhood of multipliers such that, for each agent $i$, there is a set of multipliers $\vec{\beta}$ in this neighborhood such that agent $i$'s resulting probability is larger than that of any other agent, and, as a consequence, such that $i$ has a resulting probability of at least $1/n$.

By successively refining the simplicization, we can make these neighborhoods arbitrarily small.
In \full{\cref{app:sperner}}{the full version}, we prove the existence of a set of exactly equalizing multipliers, which follows from the Bolzano-Weierstraß Theorem and continuity of the functions $p_i$ on $\mathbb{R}_{>0}^n$\full{\ (\cref{app:continuous})}{}.
\end{proof}

\subsubsection{An Approximation Algorithm for Equalizing Multipliers}
The proof above is succinct, but not particularly helpful in finding equalizing multipliers computationally.\footnote{When measuring running time, we assume that the algorithm has access to an oracle allowing it to compute the $p_i(\vec{\beta})$ for a given $\vec{\beta}$ in constant time. This choice abstracts away from the distribution-specific cost and accuracy of computing the integral in \cref{eq:piint}.}
Though the application of Sperner's lemma can be turned into an approximation algorithm, using it to find multipliers such that all resulting probabilities lie within $\delta > 0$ of $1/n$ requires $\textit{poly}(n) \, n! \, \left(\log(2 \, q) \, \left(1 + \frac{4 \, n \, q}{\delta}\right)\right)^n$ time (\full{\cref{app:sperner}}{see full version}).
This large runtime complexity points to a more philosophical shortcoming of our proof of \cref{thm:equalizing}, namely, that it does very little to elucidate the structure of how multipliers map to resulting probabilities.
Given that the proof barely made use of any properties of the $p_i$ other than continuity, it is natural that the resulting algorithm resembles a complete search over the space of multipliers.

By contrast, our polynomial-time algorithm for finding approximately-equalizing multipliers will be based on three structural properties of the $p_i$ (proved in \full{\cref{app:boundmult}}{the full version}):
\paragraph{Local monotonicity:} If we change the multipliers from $\vec{\beta}$ to $\vec{\beta}'$, and if agent $i$'s multiplier increases by the largest factor ($\beta_i'/\beta_i = \max_{j \in N} \beta_j' / \beta_j$), then $i$'s resulting probability weakly increases\full{\ (\cref{lem:localmonotonicity})}{}.
\paragraph{Bounded probability change:} If we change a set of multipliers by increasing $i$'s multiplier by a factor of $(1 + \epsilon)$ for some $\epsilon > 0$ and leaving all other multipliers equal, then $i$'s resulting probability increases by at most $2 \, q \, \epsilon$\full{\ (\cref{lem:lipschitzmultipliers})}{}.
\paragraph{Bound on multipliers:} If $i$'s multiplier is at least $2\,q$ times as large as $j$'s multiplier, then $i$ must have a strictly larger resulting probability than $j$\full{\ (\cref{lem:boundmult})}{}.\medskip

\noindent Crucially, we can combine the first two properties to control how the resulting probabilities evolve while changing the multipliers in a specific ``step'' operation, which is the key building block of our approximation algorithm:
\paragraph{Step guarantee:} If we change a set of multipliers by increasing the multipliers of a subset $S$ of agents by a factor of $(1 + \epsilon)$ while leaving the other multipliers unchanged, then (a) the resulting probabilities of all $i \in S$ weakly increase, but at most by $2 \, q \, \epsilon$, and (b) the resulting probabilties of all $i \notin S$ weakly decrease, also by at most $2 \, q \, \epsilon$ (\full{\cref{lem:smallstep}}{proof in full version}).\medskip

\begin{algorithm}[tb]
\SetKwInOut{Input}{Input} 
\SetKwInOut{Output}{Output} 
\Input{An oracle to compute resulting probabilities, a constant $0 < \delta \leq 1$, and a PDF upper bound $q$}
\Output{A vector of multipliers $\vec{\beta} \in \mathbb{R}_{>0}^n$}
$\epsilon \leftarrow \delta / (2 \, q)$\;
$\vec{z} \leftarrow \vec{0}$\;
\While{$\exists i \in N.\; \left|p_i\big((1+\epsilon)^{\vec{z}}\big) - 1/n\right| > \delta$}{
$S \leftarrow \{i \in N \mid p_i\big((1+\epsilon)^{\vec{z}}\big) \leq 1/n\}$\;
$\vec{z} \leftarrow \vec{z} + \bone_S$\;
}
\Return $(1+\epsilon)^{\vec{z}}$
\caption{Equalizing Multipliers}
\label{alg:multipliers}
\end{algorithm}

\noindent Algorithm~\ref{alg:multipliers} keeps track of a set of multipliers $(1 + \epsilon)^{\vec{z}} = ((1 + \epsilon)^{z_1}, \dots, (1 + \epsilon)^{z_n})^\mathsf{T}$.
In each loop iteration, we use the step operation to increase all resulting probabilities originally below $1/n$ and decrease all resulting probabilities originally above $1/n$, both by a bounded amount so that they cannot overshoot $1/n$ by too much.
After polynomially many steps, all resulting probabilities lie within a band around $1/n$, which means that the multipliers are approximately equalizing.

\begin{theorem}
\label{thm:equalizingpoly}
In time $\mathcal{O}(n^2 \, q \, \log (q) \, \delta^{-1})$, Algorithm~\ref{alg:multipliers} computes a vector of multipliers $\vec{\beta}$ such that, for all $i \in N$,
$ 1/n - \delta \leq p_i(\vec{\beta}) \leq 1/n + \delta$.
\end{theorem}
\begin{proof}
At the beginning of an iteration of the loop, partition the agents into three sets $Z_{\ell}$, $Z_{m}$, and $Z_{h}$ depending on whether $p_i\big((1+\epsilon)^{\vec{z}}\big)$ is smaller than $1/n - \delta$, is in $[1/n - \delta, 1/n + \delta]$, or is larger than $1/n + \delta$, respectively.
We make two observations: \textbf{(a)} Once an agent is in $Z_m$, they will always stay there since, by the step guarantee, their probability moves by at most $\delta=2 \, q \, \epsilon$ per iteration and moves up whenever it was below $1/n$ and down whenever it was above $1/n$. \textbf{(b)} Agents cannot move between $Z_\ell$ and $Z_h$ within one iteration, since the probabilities belonging to $Z_\ell$ and $Z_h$ are separated by a gap of size $2\,\delta$, whereas the step guarantee shows that an agents' probability moves by at most $\delta$. 

Next, we show that the algorithm terminates; specifically, that it exits the loop after at most $T \coloneqq \lceil\frac{\log(2 \, q)}{\log(1 + \epsilon)}\rceil \cdot (n - 1)$ iterations.
For the sake of contradiction, suppose that at the beginning of the $(T+1)$th iteration of the loop, some agent was not yet in $Z_m$.
For now, say that such an $i$ was in $Z_h$ and let $i$ have maximal $p_i\big((1+\epsilon)^{\vec{z}}\big)$.
Then, since $i$ has always been in $Z_h$, $z_i$ has never been increased and it still holds that $z_i = 0$.
At the same time, in each round, the multipliers of some $|S| \geq 1$ other agents get increased, from which it follows that some other agent $j$ must have $z_j \geq \lceil\frac{\log(2 \, q)}{\log(1 + \epsilon)}\rceil$.
Then, $\beta_j / \beta_i = (1 + \epsilon)^{z_j-z_i} \geq 2\,q$, which implies that $p_j > p_i$ by the bounds-on-multipliers property, which contradicts our choice of $i$.
The case where $i \in Z_\ell$ is symmetric:
This time, choose an $i$ with minimal $p_i \, \big((1+\epsilon)^{\vec{z}}\big)$.
Since $i \in Z_\ell$, $z_i$ must have been increased in every round and equal $T$.
Furthermore, in each previous round, $|S| \leq n-1$, since the algorithm would not have re-entered the loop if all probabilities were $1/n$, suggesting that some agent's probability is larger than $1/n$ and is thus not included in $S$.
Hence, there must be another agent $j$ with $z_j \leq T - \lceil\frac{\log(2 \, q)}{\log(1 + \epsilon)}\rceil$.
This implies that $p_j < p_i$, contradicting our choice of $i$.

It follows that the loop is executed at most $T$ times. Taking into account that each iteration requires $\mathcal{O}(n)$ oracle queries, the total time complexity is in
    $T \cdot \mathcal{O}(n) \in \mathcal{O}(n^2 \, q \, \log(q)/\delta)$.\footnote{$T \cdot \mathcal{O}(n) = \lceil \log(2 \, q) / \log(1 + \delta /(2 \,q))\rceil \, (n-1) \, \mathcal{O}(n)
    \leq \lceil \log(2 \, q) \, ((4 \, q)/\delta) \rceil \, (n-1) \, \mathcal{O}(n) \in \mathcal{O}(n^2 \, q \, \log(q)/\delta)$, where the inequality holds since $\delta/(2\,q)\leq 1$.}
The bound on the resulting probabilities follows from the fact that $Z_m = N$ when the algorithm exits.
\end{proof}
As another demonstration of the rich structure in the $p_i$, we show in \full{\cref{app:unique}}{the full version} that the equalizing multipliers are unique, using a strengthened local-monotonicity property.
\full{\new{The algorithmic proof above yields an alternative proof of the existence of perfectly equalizing multipliers, by applying a limit argument similar to the one at the end of \cref{thm:equalizing}.}}{}

\subsection{Gap between Expected Utilities}
\label{sec:existence_gap}
Having established the existence of (approximately) equalizing multipliers $\vec{\beta}$, we will now analyze the corresponding multiplier allocation, which assigns each item $\alpha$ to the agent $i$ with maximal $\beta_i \, u_i(\alpha)$.
By definition, this allocation satisfies fPO and thus PO, so it remains to show EF.
\new{In our exposition, we will focus on exactly equalizing multipliers, but all observations extend to multipliers that are ``sufficiently close'' to equalizing, which we make explicitly in \cref{lem:boundgap}.}

As sketched in \cref{sec:results}, we now prove that an agent $i$'s expected utility for an item they receive themselves is strictly larger than $i$'s expected utility for an item that another agent receives in the multiplier allocation.
In fact, we will prove that there is a \emph{constant} gap between these conditional expectations i.e., a constant $C_{p,q} > 0$ such that, for all $i \neq j \in N$, $\mathbb{E}\left[u_i\, |\, \beta_iu_i=\max_{k\in N} \beta_k u_k \right] \geq C_{p, q} + \mathbb{E}\left[u_i\, |\, \beta_ju_j=\max_{k\in N} \beta_k u_k \right]$.
Bounding this gap is the main idea of the proof by \citeauthor{DGK+14}
Their proof approach is applicable since, by scaling the utilities by equalizing multipliers, we bring a key property exploited by \citeauthor{DGK+14}\ to the asymmetric model: as does the welfare-maximizing algorithm in the symmetric setting, the multiplier allocation gives a random item to each agent with equal probability.
Thus, by concentration, all agents receive similar numbers of items.
A positive gap $C_{p,q}$ furthermore ensures that agents prefer the average item in their own bundle to the average item in another bundle.
The last two statements imply that the allocation is likely to be envy-free.

\iffullversion
Before we go into the bounds, it is instructive to see why the interval support property is required for the multiplier approach.
Consider a case with two agents: Agent A's utility is uniformly distributed on $[1/4, 3/4]$, whereas agent B's distribution is uniform on $[0,1/4]\cup[3/4, 1]$, which means that $\mathcal{D}_B$'s support is not an interval.
It is easy to verify that setting both multipliers to 1 is equalizing.
But $\mathbb{E}[u_A\, |\, u_A\geq u_B] = \mathbb{E}[u_A]$, since the event $u_A\geq u_B$ only tells us that $u_B$ is taken from the left interval in its support ($[0,1/4]$) but $u_A$ is still distributed uniformly in $[1/4, 3/4]$.
Hence, without assuming interval support, the gap we aim to bound may be zero.

For a fixed set of distributions, interval support is enough to provide a positive gap (\full{\cref{app:positivegap}}{see full version}).
However, if we want to add new agents along the infinite sequence of instances as $m \to \infty$, we require a uniform lower bound on this gap.
In \full{\cref{app:pdfboundeddisc}}{the full version}, we \full{}{also\ }give examples showing that both very high probability densities and very low probability densities can make the gap arbitrarily small, which motivates our assumption of $(p, q)$-PDF-boundedness.
\fi
In \full{\cref{app:boundgap}}{the full version}, we derive the desired\full{\ constant}{} gap:
\begin{restatable}{proposition}{lemboundgap}
\label{lem:boundgap}
For any collection of agents whose utility distributions are $(p,q)$-PDF-bounded and have interval support, given a set of multipliers $\vec{\beta}$ such that $|p_i-1/n|<1/(2\,n)$ for all $i\in N$, it holds that for any $i\ne j\in N$,
\begin{align*}
    &\mathbb{E}\left[u_i \bigg| \beta_i u_i = \max_{k\in N} \beta_k u_k \right]
    - \mathbb{E}\left[u_i \bigg| \beta_j u_j = \max_{k\in N} \beta_k u_k \right] \\
    &\geq C_{p, q}
\end{align*}
for a constant $C_{p, q} \in (0, 1]$ that only depends on $p$ and $q$.
\end{restatable}
\iffullversion
\fi
\subsection{The Multiplier Allocation Satisfies EF}
\label{sec:puttingtogetheref}
In \full{\cref{app:combine}}{the full version}, we combine the existence of equalizing multipliers and the positive gap to prove \cref{thm:existence}, i.e., that the multiplier allocation is EF with high probability, and that this even holds when assigning based on approximately equalizing multipliers.
Here, we sketch the argument: First, we run Algorithm~\ref{alg:multipliers} to find approximately equalizing multipliers $\vec{\beta}$ with an accuracy $\delta \coloneqq C_{p, q}/(4\,n)$, which requires $\mathcal{O}(n^2/\delta) \subseteq \mathcal{O}(n^3)$ time by \cref{thm:equalizingpoly}.
Second, we allocate all items based on $\vec{\beta}$, in $\mathcal{O}(m \, n)$ time).
This yields the \emph{approximate multiplier allocation}, which is fPO by construction.
It remains to show EF:
\cref{lem:boundgap} applies to $\vec{\beta}$ since it satisfies the proposition's precondition: $|p_i - 1/n| \leq \delta \leq \frac{1}{4 \, n} < \frac{1}{2 \, n}$.
By arithmetic, the positive gap guaranteed by the proposition implies that, for any two agents $i \neq j$,
    $\mathbb{E}[u_i(A_i)] -\mathbb{E}[u_i(A_j)] \geq \frac{m}{2 \,n} \, C_{p,q}$.
Assuming that $m \in \Omega(n \, \log n)$, we prove that, with high probability, all $u_i(A_j)$ stay within a distance of $\frac{m}{4n} C_{p,q}$ from their expectations by concentration. Then,
    \begin{align*}
        u_i(A_i) &\geq \mathbb{E}[u_i(A_i)] - \tfrac{m}{4n} C_{p,q} \\
        &\geq \big(\mathbb{E}[u_i(A_j)] + \tfrac{m}{2n} C_{p,q}\big) - \tfrac{m}{4n} C_{p,q} \\
        &\geq \mathbb{E}[u_i(A_j)] + \tfrac{m}{4n} C_{p,q} \geq u_i(A_j),
    \end{align*}
which implies that $i$ does not envy $j$ for any $i$ and $j$, i.e., the allocation is envy-free (and Pareto optimal) as claimed.

\section{Empirical Results}
\label{sec:experiments}
\begin{figure}
    \centering
    \includegraphics[width=.97\linewidth]{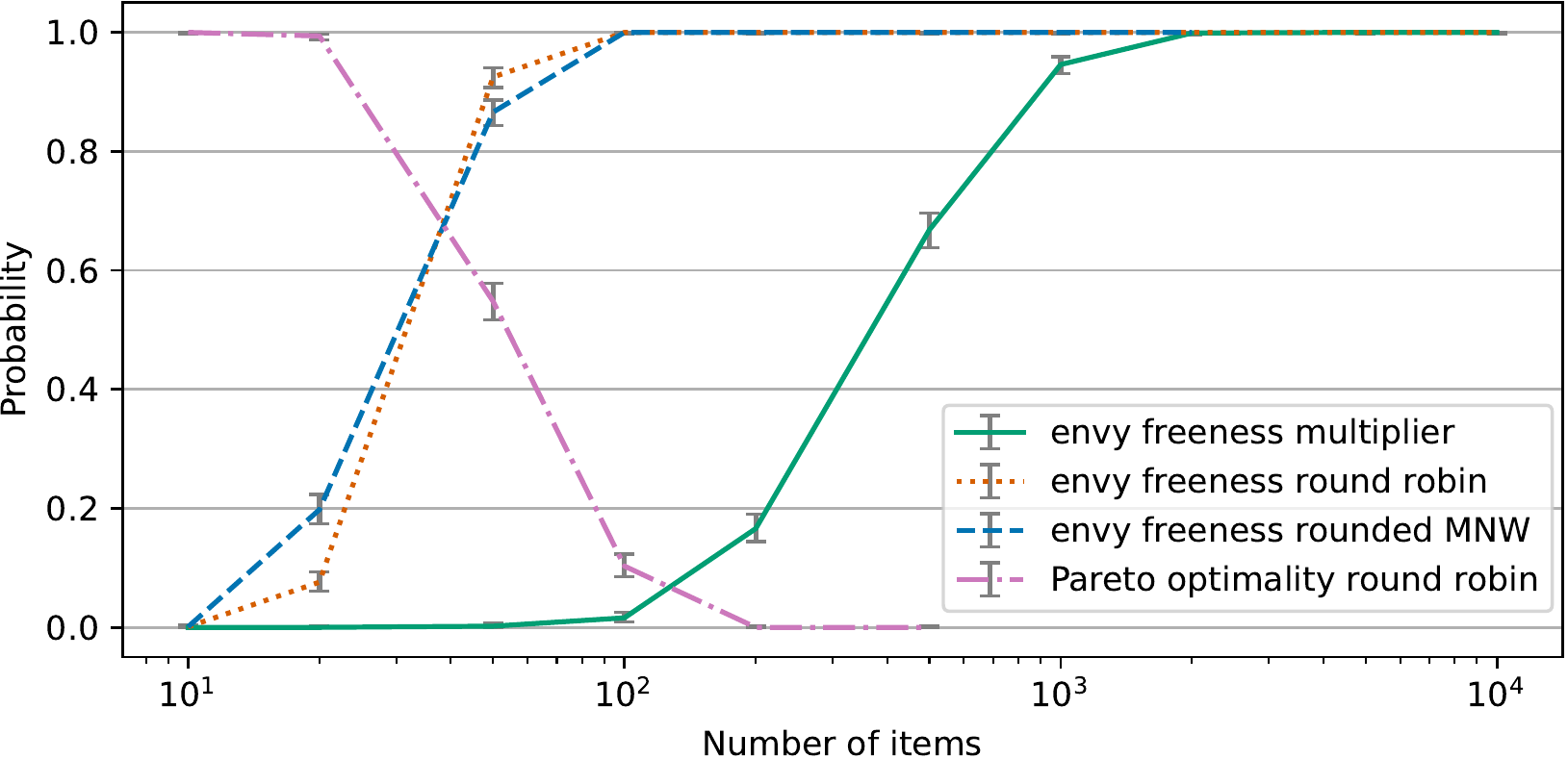}
    \caption{Probability of different algorithms satisfying EF and PO for $n=10$ distributions.
    The multiplier and rounded MNW algorithms are always PO and therefore not shown.
    Each datapoint corresponds to 1\,000 random instances, and flies indicate 95\% confidence intervals. PO and MNW data not available for large $m$ due to computational cost.}
    \label{fig:experiments}
\end{figure}
After characterizing the existence of EF and PO allocations from an asymptotic angle, we now empirically investigate allocation problems for a concrete set of agents.
We use a set of ten agents with utility distributions from a simple parametric family of $(0.1, 1.9)$-PDF-bounded distributions.\footnote{\full{\Cref{app:experiment1}}{The full version} contains all details on the experiments.
Since EF allocations exist for smaller $m$ if $n$ divides $m$, we repeat the experiment with shifted values of $m$\full{\ in \cref{app:experiment2}}{}, which does not change the major trends.
We also repeat the experiments with the five distributions from \cref{fig:scaling}, showing that our observations generalize to extremely heterogeneous distributions that are not $(p,q)$-PDF-bounded.}
We compute multipliers for these ten distributions by implementing a variant of Algorithm~\ref{alg:multipliers}.
Specifically, we repeatedly run the algorithm with exponential decreasing $\delta$, starting each iteration from the last set of multipliers, which allows the algorithm to change multipliers faster in the first rounds and empirically leads to a sublinear running time in $\delta^{-1}$.
For a requested accuracy of $\delta = 10^{-5}$, our algorithm runs in 30 seconds on consumer hardware, and we verify analytically that the resulting multipliers indeed lie within this tolerance, undisturbed by numerical inaccuracies in the computation.

As shown by the solid line in \cref{fig:experiments}, the multiplier allocation requires large numbers of items to be reliably EF:
When allocating $m=500$ items to the ten agents, the allocation is EF in only 67\% of instances, and it requires $m=2\,000$ items for this probability to reach 99\%.
This slow speed of convergence seems to be inherent to the argument of \citet{DGK+14} since, in an instance with ten copies of one of our distributions and 500 random items, the welfare-maximizing algorithm is also only EF with 87\% probability.

In contrast to the approximate-multiplier algorithm, the round robin algorithm (dotted line) reliably obtains EF allocations already for $m \geq 100$, but its allocations are essentially never PO unless $m$ is very small (dash-dotted line). This lack of PO matches our theoretical predictions \full{(\cref{app:neg})}{in the full version}.

If one searches for an algorithm that satisfies both EF and PO for small numbers of items, variants of the Maximum Nash Welfare (MNW) algorithm appear promising in our experiments\,---\,unless their computational complexity is prohibitive.
In experiments in \full{\cref{app:betaexperiments}}{the full version} with only five agents, the optimization library BARON can reliably find the (discrete) MNW allocation for small $m \leq 200$.
The MNW allocation is automatically PO, and it satisfies EF as reliably as round robin in our experiments.
For our 10 agents, however, and as little as $m=20$ items, BARON often takes multiple minutes to compute a single allocation, making this algorithm intractable for our analysis.
In \full{\cref{app:experiment1}}{the full version}, we discuss approaches to this intractability, and propose to round the \emph{fractional} MNW allocation instead.
This approach is still guaranteed to be PO, and yields EF allocations already for small $m$ (dashed line).
In a sense, this rounded MNW algorithm complements the approximate multiplier algorithm:
For small $m$, rounded MNW already provides EF and its runtime is acceptable. For large $m$, the approximate multiplier algorithm guarantees EF while its runtime scales blazingly fast in $m$, since almost all work happens in the determination of the multipliers, independently of $m$.

\section{Discussion}
In this paper, we show that EF and PO allocations are likely to exist for random utilities even if different agents' utilities follow different distributions.
Given that the known asymptotic bounds for the existence of EF+PO allocations are equal in the asymmetric and in the symmetric model, we see no evidence that the asymmetry of agent utilities would make EF+PO allocations substantially rarer to exist, up to, possibly, a $\log \log n$ gap that remains open in both models.

The most interesting idea coming out of this paper is the technique of finding equalizing multipliers, which might be of use in wider settings.
Notably, the existence proof based on Sperner's lemma mainly uses the continuity of the function mapping multipliers to probabilities, and in particular does not use the independence between the agents' utilities.
Thus, the multiplier technique might apply to random models where the agents' utilities exhibit some correlation, as long as the gap in expected utilities can still be bounded.
In the limit of infinitely many items, we can think of the multiplier technique as a way to find an allocation of divisible goods that is Pareto-optimal and \emph{balanced}, i.e., where every agent receives an equal amount of items.
In future work, we hope to explore if this construction extends to arbitrary sets of divisible items.

\section*{Acknowledgments}
We thank Bailey Flanigan and Ariel Procaccia
for valuable comments and suggestions on the paper. We also
thank Dravyansh Sharma, Jamie Tucker-Foltz, Ruixiao Yang,
Xingcheng Yao, and Zizhao Zhang for helpful technical discussions.

\bibliographystyle{named}
\bibliography{ijcai22}

\iffullversion
\newpage
\onecolumn
\appendix

\section*{\LARGE{Appendix}}
\section{Proof of \cref{prop:rr}: Existence of Envy-Free Allocations}
\label{app:rr}
\proprr*
\begin{proof}
For any pair of agents $i, j\in N$, we will prove the probability of the event that $i$ envies $j$ in our asymmetric model is in $O\left(1/m^3\right)$.
Consider a symmetric model with distribution $\mathcal{D}_i$ for all agents. Let $\mathbb{P}^S\left[T\right]$ be the probability that some event $T$ occurs in this symmetric model, and $\mathbb{P}^A\left[T\right]$ be the probability that some event $T$ occurs in the asymmetric model.

For every utility profile, define its ordinal profile as $\{\mathcal{O}_k\}_{k\in N}$ where $\mathcal{O}_k$ is a permutation of items $M$, in descending order according to $u_k(\alpha)$ for $\alpha\in M$.
Let $\mathbb{O}$ be the set of all possible ordinal profiles, which contains $(m!)^n$ elements. Since in both models, all agents are independent and the utility for all items are drawn independently, all ordinal profiles have the same probability to appear, that is, 
\begin{align*}
    \forall\, \Tilde{O}\in \mathbb{O}, \,
    \mathbb{P}^S\left[\{\mathcal{O}_k\}_{k\in N}=\Tilde{O}\right] = \mathbb{P}^A\left[\{\mathcal{O}_k\}_{k\in N}=\Tilde{O}\right] = \frac{1}{(m!)^n}.
\end{align*}
Since the allocation allocated by round robin algorithm is uniquely determined by the ordinal profile, let $\{A_k(\Tilde{O})\}_{k\in N}$ denote the resulting allocation given ordinal profile $\Tilde{O}$.
We can express $\mathbb{P}^S\left[i\text{ envies }j\right]$ as
\begin{align*}
    \mathbb{P}^S\left[i\text{ envies }j\right]  &=\sum_{\Tilde{O}\in\mathbb{O}}\mathbb{P}^S\left[\{\mathcal{O}_k\}_{k\in N}=\Tilde{O}\right]\cdot\mathbb{P}^S\left[i\text{ envies }j \, \bigg| \, \{\mathcal{O}_k\}_{k\in N}=\Tilde{O}\right] \\
    &= \frac{1}{(m!)^n}\sum_{\Tilde{O}\in\mathbb{O}}\mathbb{P}^S\left[i\text{ envies }j\, \bigg|\, \{\mathcal{O}_k\}_{k\in N}=\Tilde{O}\right].
\end{align*}
Similarly, we can express $\mathbb{P}^A\left[i\text{ envies }j\right]$ as
\begin{align*}
    \mathbb{P}^A\left[i\text{ envies }j\right] 
    = \frac{1}{(m!)^n}\sum_{\Tilde{O}\in\mathbb{O}}\mathbb{P}^A\left[i\text{ envies }j\, \bigg|\, \{\mathcal{O}_k\}_{k\in N}=\Tilde{O}\right].
\end{align*}
For all $\Tilde{O}\in \mathbb{O}$, 
\begin{align*}
    \mathbb{P}^S\left[i\text{ envies }j\, \bigg|\, \{\mathcal{O}_k\}_{k\in N}=\Tilde{O}\right] 
    &= \mathbb{P}^S\left[\{u_i(\alpha)\sim \mathcal{D}_i\}_{\alpha\in M}: \sum_{\alpha\in A_i(\Tilde{O})}u_i(\alpha)<\sum_{\alpha\in A_j(\Tilde{O})}u_i(\alpha)\, \bigg|\,\mathcal{O}_i=\Tilde{O}_i\right] \\
    &= \mathbb{P}^A\left[i\text{ envies }j\, \bigg|\, \{\mathcal{O}_k\}_{k\in N}=\Tilde{O}\right].
\end{align*}
Hence we have $\mathbb{P}^A\left[i\text{ envies }j\right]=\mathbb{P}^S\left[i\text{ envies }j\right]$.
The following lemma is implied by the proof for Thm 3.1 by \citet{MS20}.
\begin{lemma}[\citealt{MS20}]
In the symmetric model, if $m\in\Omega\left(n\,\log n/\log\log n\right)$ and the common distribution $\mathcal{D}$ is $(p, q)$-PDF-bounded on $[0, 1]$, then for any pair of agents $i, i'$, the probability that $i$ envies $i'$ in the round robin allocation is at most $O(1/m^3)$.
\label{lem:MSenvy}
\end{lemma}
\paragraph{PDF-bounded on $\mathbf{[0, 1]}$}
Here we follow the assumptions of \citet{MS20} on the distributions: PDF-bounded on $[0,1]$.
Since $\mathcal{D}_i$ is $(p, q)$-PDF-bounded, the lemma indicates that $\mathbb{P}^S\left[i\text{ envies }j\right]=O\left(1/m^3\right)$. Thus, by our earlier arguments, the probability that agent $i$ envies agent $j$ in our asymmetric model is also in $O\left(1/m^3\right)$.
Applying a union bound over all pairs $i, j$, we know that the allocation is envy-free in the asymmetric model with probability at least $1-O\left(1/m\right)$ when $m\in\Omega\left(n\,\log n/\log\log n\right)$.

\paragraph{Interval support and PDF-bounded}
Moreover, we make slight modification (on constant level) to the proof by \citet{MS20} to generalize \cref{lem:MSenvy} and get bounded envy probability when only assuming interval support and $(p,q)$-PDF-bounded. 

We first review the main idea of the proof for Thm 3.1 in their paper. 
For two agents $i, i'$, let $X^{i, i'}_t$ denote $i$'s value for the item that $i'$ gets in the $t$th round, and $X^{i}_t$ denote $i$'s value for her own item in the $t$th round. While the maximum possible envy can be (when $i'$ chooses before $i$ in each round and gets 1 more item than $i$)
\begin{equation}
    u_i(M_{i'}) - u_i(M_{i}) = X_1^{i,i'}-\sum_{t}\left(X^i_t-X^{i,i'}_{t+1}\right)
    \leq 1-\sum_{t=1}^{T}X^i_t\cdot\left(1-Y^{i,i'}_{t+1}\right).
\label{eq:ms21}
\end{equation}
where the last inequality follows from $Y^{i,i'}_{t+1}<1$ for all $t\geq1$, the gap in the first $T$ rounds is sufficient for the envy to be negative with high probability.\footnote{Note that we allow negative envy, whereas some works define envy to be $\max\{0, u_i(A_j)-u_i(A_i)\}$. When envy is $-e$, we can also say the negative envy is $e$.} \citeauthor{MS20} choose such $T$ that with high probability ($1-O\left(1/m^3\right)$), events 
(E1) $\sum_{t=1}^{T}Y_t^{i,i'}\geq T-2$ and 
(E2) $X^i_t<1/2$ for all $t\leq T$, do not happen.
Then it can be seen that when neither E1 nor E2 happen, the envy in \cref{eq:ms21} is non-positive.

For constant $c>1$, consider changing E1 to E1':
(E1') $\sum_{t=1}^{T}Y_t^{i,i'}\geq T-2c$, then E1' and E2 give us negative envy of at least $c-1$.
We can still bound the probability that E1' or E2 occur in $O\left(1/m^3\right)$, by multiplying the value of $T$ set in Manurangsi and Suksompong's proof by a factor of $c$, while keeping other valuations as they did.
Then by Lemma 2.4 in their paper, the upper bound of the probability that E1' occurs is the upper bound for E1 to the power of $c$, which is still in $O\left(1/m^3\right)$. Meanwhile, the upper bound of the probability that E2 occurs is still in $O\left(1/m^3\right)$, for $m$ that is sufficiently large.
Hence we show that for some constant $c>1$, when the distribution is PDF-bounded on $[0,1]$, the probability that $i$ envies $i'$ more than $1-c$ in the round robin allocation is at most $O\left(1/m^3\right)$.

Now we use such result to further prove the envy probability is bounded by $O\left(1/m^3\right)$ when the distribution $\mathcal{D}$ is PDF-bounded and has interval support instead of $[0,1]$ support.
The method is to use affine transformation to transform $\mathcal{D}$'s support interval $[a, b]$ into $[0, 1]$, mapping the original utility $u$ to $u'=(u-a)/(b-a)$ and the original distribution $\mathcal{D}$ to $\mathcal{D}'$. The PDF $f_\mathcal{D}$ now becomes $f_\mathcal{D'}(u)=(b-a)\cdot f_\mathcal{D}((b-a)u+a)$. 
Since $\mathcal{D}$ is $(p, q)$-PDF-bounded, the length of its support, $b-a$, must be at least $1/q$. Thus for any $u\in[0, 1]$, $p/q\leq f_\mathcal{D'}(u)\leq q$, indicating that $\mathcal{D}'$ is PDF-bounded on $[0, 1]$.
Since the affine transform does not change the ordinal profile, the round robin allocation under the transformed utility profile, where all original utilities are transformed by the affine transformation: $u\mapsto (u-a)/(a-b)$, is the same as the one under the original utility profile. 
For the same round robin allocation, suppose the envy that $i$ holds for $i'$ in the transformed utility profile is $e$, then the envy in the original utility profile becomes $e'$:
\begin{displaymath}
    e' = \begin{cases}
    (b-a)\, e & \text{if $i$ and $i'$ get same number of items}\\
    (b-a)\, e+a & \text{if $i'$ gets one more item than $i$} \\
    (b-a)\, e-a & \text{if $i$ gets one more item than $i'$}
    \end{cases}
\end{displaymath}
Then for $i$ to envy $i'$ in the original utility profile, i.e., $e'>0$, it must be true that $e>-a/(b-a)$. Since the utilities in the transformed utility profile can be considered as drawn randomly from $\mathcal{D}'$, which is PDF-bounded on $[0, 1]$, by our previous result, the probability that $e>-a/(b-a)=1-c$ is in $O\left(1/m^3\right)$.
Hence the probability that $i$ envies $i'$ in the round robin allocation for distribution $\mathcal{D}$ is still in $O\left(1/m^3\right)$, generalizing the result \cref{lem:MSenvy} to only assuming interval support and PDF-bounded for the distributions.
Finally, similarly, we get $\mathbb{P}^S\left[i\text{ envies }j\right]=O\left(1/m^3\right)$ and that the allocation is envy-free in the asymmetric model with high probability, when distributions have interval support and are $(p,q)$-PDF-bounded.

\end{proof}

\section{Discussion of Maximum-Percentile Algorithm}
\label{app:maxpercentile}
\begin{figure}
    \centering
    \includegraphics[width=.3\linewidth]{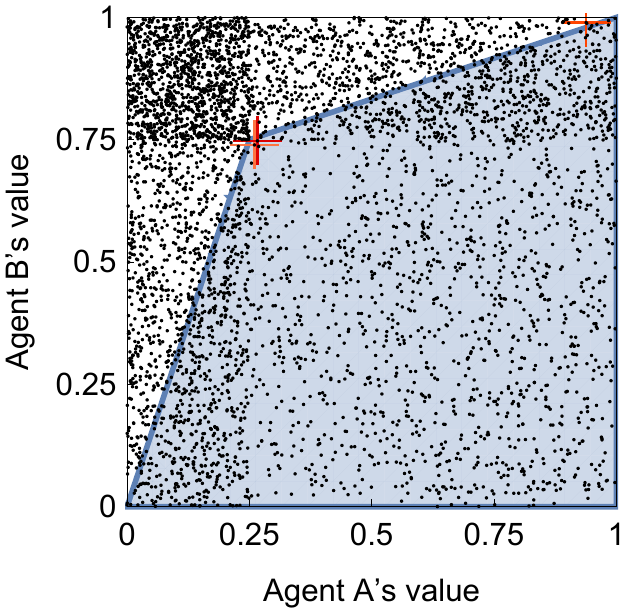}
    \includegraphics[width=.3\linewidth]{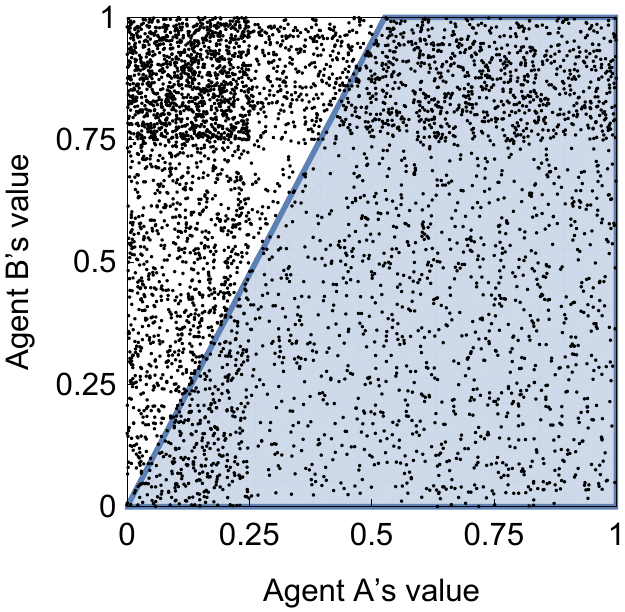}
    \caption{Illustrations of the example in the text. Dots represent utilities of 5\,000 random items. Shaded regions delineate items allocated to agent~A, for the maximum-percentile algorithm (left) and for the multiplier allocation (right). Three marked items (two have near-identical utilities) show that the maximum-percentile allocation is not PO.}
    \label{fig:percentile}
\end{figure}

As we stated in the body of the paper, \citet{KPW16} apply the proof by \citet{DGK+14} to show the existence of EF allocations in the asymmetric setting.
The core idea of their algorithm is to allocate each item to the agent for whom the item is in the highest percentile of their utility distribution, which we will call the \emph{maximum-percentile algorithm}.
It is easy to see that each agent has a probability $1/n$ of receiving each item, and it is not too hard to show that agents have higher expected utility for items they receive than for items allocated to other agents. This implies envy-freeness with high probability as $m \in \Omega\left(n\, \log n\right)$ following the proof by \citeauthor{DGK+14}

Unfortunately, this construction is unlikely to generate PO allocations:
Consider a setting with two asymmetric agents, in which agent A's utility is drawn, with 50\% probability, uniformly between 0 and $1/4$, and, with 50\% probability, uniformly between $1/4$ and 1; and in which agent B's utility is drawn either uniformly between $0$ and $3/4$ or uniformly between $3/4$ and 1, each with 50\% probability.
Conceptually, A's utility distribution skews towards lower values, whereas B's skews towards higher values.
The black dots in \cref{fig:percentile} show random samples of these utilities, and the shaded region in the left plot marks the range of utilities in which the maximum-percentile algorithm allocates items to A.
The left plot also highlights three specific items: two lie around the median utility for both agents and are given to A, and one of them lies around the top percentile for both agents and is given to B.
The fact that the ratio ``$u_B(\alpha)/u_A(\alpha)$'' is strictly greater for the two ``median items'' (at roughly $(3/4)/(1/4) = 3$) than for the ``top item'' (roughly $1$) immediately implies that the allocation is not fPO: Agent B would profit from trading half of the top item against one of A's median items (roughly, since $3/4 > 1/2$), and A would also profit from this trade (since $1/2 > 1/4$).
In fact, a similar trade of whole items, which exchanges both median items against the top item, shows that the maximum-percentile allocation violates Pareto-optimality proper.

The most promising way to avoid violations of PO is to construct fPO allocations, since the characterization of fPO using multipliers in \cref{sec:model} provides useful structure that is not available for PO.
As shown in the right panel in \cref{fig:percentile}, this corresponds to choosing a line through the origin, allocating items below the line to A, and allocating items above the line to B.
In fact, the plot shows the unique such line with the added property that a random-utility item is equally likely to be given to either agent.
In the next section, we generalize this kind of allocation to arbitrary numbers of agents.

\section{Discussion of Normalizing Multipliers}
\label{app:normalize}
As discussed in previous section, the most promising way to construct PO allocation is to utilize multiplier-based maximum-welfare allocation.
One natural choice is the set of normalizing multipliers that normalize each agent's values to add up to 1.
However, we show by the following counter-example that the set of normalizing multipliers may violate EF.

Consider a setting with $|N|=n$ agents where agent 1's utility is drawn uniformly from $[0, 1]$ while the rest of the agents' utilities are drawn uniformly from $[1/3, 2/3]$. When the number of items is sufficiently large, the normalizing multipliers for all agents will be close. With dominating probability, by Chernoff bound, all multipliers for agents $N/\{1\}$ are less than $5/4$ of agent 1's multiplier. A Chernoff bound will also guarantee high probability that agent 1 receives more than $1/12=1/2\, (1-2/3\cdot5/4)$ of all items, while there must be some agent $i$ receiving less than $1/n$ of all items. Agent $i$ will surely envy agent 1 when $1/n\cdot2/3<1/12\cdot1/3$.

\section{Proofs Used in Existence Result of Envy-free and Pareto-optimal Allocations}

\subsection{Continuity of Probability Function $p_i(\cdot)$}
\label{app:continuous}
\begin{lemma}
\label{thm:continuous}
For all $i\in N$, the probability function $p_i(\cdot)$ defined by 
\begin{align*}
    p_i\left(\beta_1, \dots, \beta_n\right) = \mathbb{P}\left[\beta_i u_i = \max_{j\in N}\beta_j u_j\right]
\end{align*}
is continuous on $(\beta_1, \dots, \beta_n)\in \mathbb{R}^n_{>0}$.
\end{lemma}

\begin{proof}
We have
\begin{align*}
    p_i(\beta_1,\dots,\beta_n) &= \int_{0}^{1} f_i(u)\, \prod_{j\ne i}F_j\left(\frac{\beta_i}{\beta_j}u\right)\, du, 
\end{align*}
then
\begin{displaymath}
\begin{aligned}
    \left|p_i\left(\beta'_1,\cdots,\beta'_n\right) - p_i(\beta_1,\cdots,\beta_n)\right| 
    &\leq \int_{0}^{1} f_i(u)\cdot
    \left|\, \prod_{j\ne i}F_j\left(\frac{\beta'_i}{\beta'_j}u\right)-\prod_{j\ne i}F_j\left(\frac{\beta_i}{\beta_j}u\right)\right|\, du \\
    &\leq \sum_{j\ne i}\int_{0}^{1} f_i(u)\cdot
    \left| F_j\left(\frac{\beta'_i}{\beta'_j}u\right)-F_j\left(\frac{\beta_i}{\beta_j}u\right)\right|\, du.
\end{aligned}
\end{displaymath}
For nonatomic distribution $\mathcal{D}_j$, its cumulative distribution function $F_j(\cdot)$ is continuous, also the function $g_u(x, y) = \frac{xu}{y}$ is continuous when $x, y\ne 0$. Thus 
\begin{align*}
    \lim_{(\beta'_i, \beta'_j)\rightarrow(\beta_i, \beta_j)}F_j\left(\frac{\beta'_i}{\beta'_j}u\right) = F_j\left(\frac{\beta_i}{\beta_j}u\right),
\end{align*}
and we have
\begin{align*}
    &\lim_{(\beta'_1, \cdots, \beta'_n)\rightarrow(\beta_1, \cdots, \beta_n)}\left|\, p_i(\beta'_1,\cdots,\beta'_n) - p_i(\beta_1,\cdots,\beta_n)\,\right| \\
    &\leq \sum_{j\ne i}\int_{0}^{1} f_i(u) \lim_{(\beta'_i, \beta'_j)\rightarrow(\beta_i, \beta_j)} \left|F_j\left(\frac{\beta'_i}{\beta'_j}u\right) - F_j\left(\frac{\beta_i}{\beta_j}u\right)\right| du = 0.
\end{align*}
Therefore function $p_i(\cdot)$ is continuous on $\mathbb{R}^n_{>0}$.

\end{proof}

\subsection{Proof for Existence of Multipliers with Sperner's Lemma}
\label{app:sperner}
Here we show a detailed proof for existence of equalizing multipliers with Sperner's Lemma.
Without loss of generality we assume all multipliers add up to 1, then $(\beta_1, \beta_2, \dots, \beta_n)$ falls in a $(n-1)$-dimensional simplex $S$. Now we define the coloring function $f: S\rightarrow N$, which maps each set of multipliers to an agent with the highest probability of having the largest scaled utility under the set of multipliers:
\begin{displaymath}
f(\beta_1, \dots, \beta_n) \in \mathop{\arg\max}_{i\in N}\mathbb{P}\left[\beta_i u_i = \max_{j\in N}\beta_j u_j\right].
\end{displaymath}
It is clear that the $n$ vertices of the simplex are colored with $n$ different ``colors'', since agent $i$ will have probability 1 of having the largest scaled utility when $\beta_i=1$ and $\beta_j=0, \forall j\ne i$.

Now we divide the simplex $S$ into smaller simplices with (at most) half the diameter of previous simplex. Let $\mathcal{S}_1$ denote the set of smaller simplices. 
Then we further divide each simplex in $\mathcal{S}_1$ into even smaller simplices with half diameter, and we denote them by $\mathcal{S}_2$. We repeat the procedure and divide the original simplex into $\mathcal{S}_1, \mathcal{S}_2, \mathcal{S}_3, \dots$ containing smaller and smaller simplices.

By \emph{Sperner's Lemma}, in each $\mathcal{S}_i\, (i\geq1)$, there are always an odd number of simplicies that are colored with $n$ colors, indicating the existence of a simplex $\bm{s_i}\in \mathcal{S}_i$ whose vertices are colored with all $n$ colors.
Now consider the sequence of such simplices: $\bm{s_1}, \bm{s_2}, \bm{s_3}, \dots$, and let the $i$th vertex be the vertex that is mapped to $i$ by $f$. We can represent each simplex with a $n^2$-dimensional vector: 
\begin{align*}
    (\bm{\beta^1}, \bm{\beta^2}, \dots, \bm{\beta^n}) = 
    (\beta^1_1, \dots, \beta^1_n; \dots; \beta^n_1, \dots, \beta^n_n)\in \mathbb{R}^{n^2}.
\end{align*}
where $\bm{\beta^i}$ is the vector of multipliers in the $i$th vertex and $\beta^i_j$ is the $j$th multiplier in the $i$th vertex. Let $\bm{s^i_j}$ denote the subvector $\bm{\beta^i}$ in $\bm{s_j}$.

Since all of them are bounded between $(0, 0, \dots, 0)$ and $(1, 1, \dots, 1)$, by \emph{Bolzano-Weierstrass Theorem}, there must be a convergent subsequence $\bm{s_{a_1}}, \bm{s_{a_2}}, \bm{s_{a_3}}, \dots$ that converges to a simplex $(\bm{\beta^{1^*}}, \bm{\beta^{2^*}}, \dots, \bm{\beta^{n^*}})$ in the space:
\begin{align*}
    \lim_{t\rightarrow\infty}\bm{s_{a_t}} = (\bm{\beta^{1^*}}, \bm{\beta^{2^*}}, \dots, \bm{\beta^{n^*}}).
\end{align*}
Since the simplicies in $\{\bm{s_i}\}_{i\geq 0}$ get smaller and smaller with a ratio of $1/2$, we have
\begin{align*}
    \lim_{t\rightarrow\infty}\bm{s^1_{a_t}} = \lim_{t\rightarrow\infty}\bm{s^2_{a_t}} = \cdots = \lim_{t\rightarrow\infty}\bm{s^n_{a_t}}.
\end{align*}
Then it can be deduced that $\bm{\beta^{1^*}} = \cdots = \bm{\beta^{n^*}} = \bm{\beta^*}$.

We argue that $\bm{\beta^*}\in \mathbb{R}^n_{>0}$. 
Since otherwise, if $\beta^*_i=0$ for some $i\in N$, then $\beta^i_i$, the $i$th multiplier in the $i$th vertex grows arbitrarily small as the sequence $\{\bm{s_{a_t}}\}$ goes.
This would lead to the probability that agent $i$ having the largest scaled utility goes to 0, contradicting the fact that agent $i$ should have the highest probability of having largest scaled utility, with probability at least $1/n$.

Now we claim that all agents have equal probability of having the largest scaled utility under $\bm{\beta^*}$. Define 
\begin{align*}
p_i(\beta_1, \dots, \beta_n) = \mathbb{P}\left[\beta_i u_i = \max_{j\in N}\beta_j u_j\right].
\end{align*}
As shown in \cref{app:continuous}, $p_i(\cdot)$ is a continuous function in $\mathbb{R}^n_{>0}$. Then for all $i\in N$, 
\begin{align*}
    \frac{1}{n}\leq \lim_{t\rightarrow\infty}p_i(\bm{s^i_{a_t}}) = p_i(\lim_{t\rightarrow\infty}\bm{s^i_{a_t}}) = p_i(\bm{\beta^*}).
\end{align*}
Since $\sum_{i=1}^{n}p_i(\bm{\beta^*})=1$, we must have
\begin{align*}
    p_1(\bm{\beta^*}) = \cdots = p_n(\bm{\beta^*}) = \frac{1}{n}.
\end{align*}
This shows that $\bm{\beta^*}$ is the set of multipliers we are looking for, hence the existence.

\subsubsection{Algorithm Based on Sperner's-Lemma Proof}
We discuss here the degree to which the Sperner's argument can be implemented as an approximation algorithm.

\begin{proposition}
\label{prop:sperneralgo}
For any $\delta > 0$, a variant of the previous Sperner's argument computes a set of multipliers $\{\beta_i\}_{i \in N}$ such that $|p_i - p_j| \leq \delta$ in time $\textit{poly}(n) \, n! \, \left(\log(2 \, q) \, \left(1 + \frac{4 \, n \, q}{\delta}\right)\right)^n$.
\end{proposition}
\begin{proof}
To make the existence proof algorithmic, a key question is how to discretize the space of multipliers into simplices.
Such a simplicization should be easy to traverse algorithmically and its simplices should describe sufficiently compact sets of multipliers such that any panchromatic simplex should allow to uniformly bound how far the probabilities are from $1/n$.
As \citet{Papadimitriou94} describes, there is already no obvious simplicization when $n=4$; for example, the regular tetrahedron cannot be partitioned into multiple regular tetrahedra.
Following Papadimitriou's discussion, we apply Sperner's lemma to a hypercube rather than a simplex:

Let $\epsilon>0$ denote a small constant, to be determined later.
Let our set of points be the grid $G \coloneqq \{ \bm{p} \in \mathbb{Z}^{n-1} \mid ||\bm{p}||_{\infty} \leq \lceil \frac{\log(2 \, q)}{\log(1+\epsilon)} \rceil\}$.
We can partition this grid first in cubelets of the shape $[0, 1]^{n-1} + \bm{p} \subseteq G$ for some $\bm{p} \in G$, and subdivide each of these cubelets into $(n-1)!$ simplices.\footnote{For an excellent exposition, see \url{https://people.csail.mit.edu/costis/6896sp10/lec6.pdf}, accessed on January 11, 2022.}
We map each point $\bm{p} = (p_1, \dots, p_{n-1})$ in the grid $G$ to a color in $[n]$;
specifically, we choose the color $\argmax_{i \in N} p_i((1+\epsilon)^{p_1}, (1 + \epsilon)^{p_2}, \dots, (1 + \epsilon)^{p_{n-1}}, 1)$ with some canonical tie breaking.
If, for some $\bm{p} \in G$ and $1 \leq i \leq n-1$, we have $p_i = - \lceil \frac{\log(2 \, q)}{\log(1+\epsilon)} \rceil$, then $i$'s multiplier is at most $\frac{1}{2\,q}$ whereas agent $n$'s multiplier is $1$; by \cref{lem:boundmult}, $n$'s probability of being the largest is strictly larger than $i$'s, which means that $\bm{p}$'s color is not $i$.
Similarly, if $p_i = \lceil \frac{\log(2 \, q)}{\log(1+\epsilon)} \rceil$, then $i$'s probability is strictly larger than $n$'s by \cref{lem:boundmult}, and therefore $\bm{p}$'s color is not $n$.
Given these observations, there must exist a panchromatic simplex, which in turn must be contained in a cubelet $C = [0,1]^{n-1} + \bm{p}$ for a $\bm{p} \in G$.
Fix the multipliers $\{\beta_i\}_{i \in N}$ by setting $\beta_n \coloneqq 1$, and $\beta_i \coloneqq (1 + \epsilon)^{p_i}$ for all $1 \leq i \leq n-1$.
Recall that some point in $C$ is colored $n$, which means that $n$ has the largest probability for the corresponding multipliers, and in particular has a probability of at least $1/n$.
Since $n$'s multiplier $\beta_n$ is equal to the multiplier of this point, and since all other multipliers $\beta_i$ are at most the corresponding multiplier at this point, it follows that $p_n(\beta_1, \dots, \beta_n) \geq 1/n$.
Similarly, fix some agent $1 \leq i \leq n - 1$.
Since some point in $C$ has color $i$, the point $\bm{p} + \bm{e_i}$ must also have color $i$, where $\bm{e_i}$ is the unit vector in dimension $i$.
It follows that $p_i(\beta_1, \dots, \beta_{i-1}, (1 + \epsilon) \, \beta_i), \beta_{i+1}, \dots, \beta_n) \geq 1/n$.
By \cref{lem:lipschitzmultipliers}, $p_i(\beta_1, \dots, \beta_n) \geq 1/n - 2 \, q \, \epsilon$.
These lower bounds also allow us to upper bound the probabilities; for each $i \in N$, $p_i(\beta_1, \dots, \beta_n) \leq 1 - (n-1) \, (1/n - 2 \, q \, \epsilon) = 1/n + 2 \, (n - 1) \, q \, \epsilon$.
If we set $\epsilon = \frac{\delta}{2 \, n \, q}$, then it holds for all $i,j \in N$ that
\[|p_i(\beta_1, \dots, \beta_n) - p_j(\beta_1, \dots, \beta_n)| \leq \delta.\]
It remains to bound the running time of this algorithm. To find the panchromatic simplex, we traverse the simplices, visiting each at most once. Within each simplex, we query the oracle to determine the color of the new vertex and perform polynomial-time work to move on to the next simplex. The bottleneck of the computation is that we might have to traverse nearly all vertices, which leads to an overall running time in
\[ \textit{poly}(n) \, (n-1)! \, \left(\frac{2 \, \log(2 \, q) + 1}{\log(1 + \frac{\delta}{2 \, n \, q})}\right)^{n-1} \leq \textit{poly}(n) \, n! \, \left((2 \, \log(2 \, q) + 1) \, \left(1/2 + \frac{2 \, n \, q}{\delta}\right)\right)^n. \qedhere \]
\end{proof}

\subsection{Proof of Properties of the $p_i$}

\label{app:boundmult}
\begin{lemma}
\label{lem:localmonotonicity}
Fix multiplier sets $\vec{\beta}, \vec{\beta}' \in \mathbb{R}_{>0}^n$ and an agent $i$.
If $\frac{\beta_i'}{\beta_i} \geq \frac{\beta_j'}{\beta_j}$ for all $j \neq i$, then $p_i(\vec{\beta}') \geq p_i(\vec{\beta})$.
\end{lemma}
\begin{proof}
In \cref{eq:piint}, observe that the $F_j$ are monotone increasing and that $\beta_i' / \beta_j' = (\beta'_i / \beta_i)/(\beta'_j / \beta_j)\cdot\beta_i / \beta_j\geq \beta_i / \beta_j$, $\forall j$.
\end{proof}

\begin{restatable}{lemma}{lipschitzmultipliers}
\label{lem:lipschitzmultipliers}
Fix a set of multipliers $\vec{\beta}$, an agent $i$, and $\epsilon > 0$.
Let $\vec{\beta}'$ denote a set of multipliers such that $\beta'_i \coloneqq (1 + \epsilon) \, \beta_i$ and $\beta_j' \coloneqq \beta_j$ for all $j \neq i$. Then, $p_i(\vec{\beta}') \leq p_i(\vec{\beta}) + 2 \, q \, \epsilon$.
\end{restatable}
\begin{proof}
For convenience, set $p_i \coloneqq p_i(\beta_1, \dots, \beta_n)$ and $p_i' \coloneqq p_i(\beta_1, \dots, \beta_{i-1}, (1 + \epsilon) \, \beta_i, \beta_{i+1}, \dots, \beta_n)$.
Moreover, define a function $\pi : \mathbb{R} \to \mathbb{R}$ such that $\pi(t) \coloneqq \prod_{j \neq i} F_j(\frac{\beta_i}{\beta_j} \, t)$.
Observe that $\pi$ is monotone increasing and its range lies within $[0,1]$.
Using these definitions,
\begin{align}
    p_i' - p_i &= \int_{0}^1 f_i(u) \, \big(\pi((1+\epsilon) \, u) - \pi(u)\big) \, du \notag \\
    &= \sum_{t=0}^{\infty} \int_{(1+\epsilon)^{-t-1}}^{(1+\epsilon)^{-t}} f_i(u) \, \big(\pi((1+\epsilon) \, u) - \pi(u)\big) \, du \notag \\
    &\leq \sum_{t=0}^{\infty} \left(\big(\pi((1+\epsilon)^{-t+1}) - \pi((1 + \epsilon)^{-t-1})\big) \, \int_{(1+\epsilon)^{-t-1}}^{(1+\epsilon)^{-t}} f_i(u) du \right) \notag \\
    &\leq \sum_{t=0}^{\infty} \left(\big(\pi((1+\epsilon)^{-t+1}) - \pi((1 + \epsilon)^{-t-1})\big) \, \int_{(1+\epsilon)^{-t-1}}^{(1+\epsilon)^{-t}} q \, du \right) \notag \\
    &= q \, \epsilon \, \sum_{t=0}^{\infty} \big(\pi((1+\epsilon)^{-t+1}) - \pi((1 + \epsilon)^{-t-1})\big) \, (1 + \epsilon)^{-t-1} \label{eq:telescope}
\end{align}
By the monotonicity of $\pi$, all coefficients $\pi((1+\epsilon)^{-t+1}) - \pi((1 + \epsilon)^{-t-1})$ are nonnegative.
Moreover, if we add up these coefficients only for the even $t$, this is a telescoping series
\begin{align*}
    \sum_{t=0,2,4,\dots} \pi((1+\epsilon)^{-t+1}) - \pi((1 + \epsilon)^{-t-1}) = \lim_{t \to \infty} \pi(1+\epsilon) - \pi((1 + \epsilon)^{-2\,t-1}) \leq 1.
\end{align*}
Thus, the even summands of \cref{eq:telescope} are a ``subconvex'' combination of the $(1+\epsilon)^{-t-1}$, and are therefore upper bounded by the largest such term:
\[ \sum_{t=0,2,4,\dots} \big(\pi((1+\epsilon)^{-t+1}) - \pi((1 + \epsilon)^{-t-1})\big) \, (1 + \epsilon)^{-t-1} \leq (1 + \epsilon)^{-1}. \]
Applying the same reasoning to the odd summands, we obtain
\[ \sum_{t=1,3,5,\dots} \big(\pi((1+\epsilon)^{-t+1}) - \pi((1 + \epsilon)^{-t-1})\big) \, (1 + \epsilon)^{-t-1} \leq (1 + \epsilon)^{-2}. \]
By plugging these last two equations into \cref{eq:telescope}, we conclude that
\[ p_i' - p_i \leq q \, \epsilon \, \big((1 + \epsilon)^{-1} + (1 + \epsilon)^{-2} \big) \leq 2 \, q \, \epsilon. \qedhere \]
\end{proof}

\begin{restatable}{lemma}{lemboundmult}
\label{lem:boundmult}
For two agents $i$ and $j$ and $\vec{\beta} \in \mathbb{R}_{>0}^n$, if $\smash{\frac{\beta_i}{\beta_j}} \geq 2 \, q$, then $p_i(\vec{\beta}) > p_j(\vec{\beta})$.
\end{restatable}
\begin{proof}
Suppose there is a pair of $i, j\in N$ such that $\beta_i/\beta_j> 2q$. Then we we have the following inequalities for the probability $p_i, p_j$ of $i$ and $j$ getting each item:
\begin{align*}
    &p_i = \mathbb{P}\left[\beta_i u_i=\max_{k\in N}\beta_k u_k\right] 
    \geq \mathbb{P}\left[\beta_i u_i=\max_{k\in N}\beta_k u_k\,\land\, \beta_i u_i\geq\beta_j\right] \\
    &= \mathbb{P}\left[\beta_iu_i\geq\beta_j\right]\cdot\mathbb{P}\left[\beta_i u_i=\max_{k\in N}\beta_k u_k\, \bigg|\, \beta_i u_i\geq\beta_j\right] 
    \geq \mathbb{P}\left[\beta_iu_i\geq\beta_j\right]\cdot\mathbb{P}\left[\beta_j\geq\max_{k\in N,\ k\ne i, j}\beta_k u_k\right],
\end{align*}
and 
\begin{align*}
    &p_j = \mathbb{P}\left[\beta_j u_j=\max_{k\in N}\beta_k u_k\right] 
    = \mathbb{P}\left[\beta_j u_j\geq \beta_i u_i\right]\cdot\mathbb{P}\left[\beta_j u_j=\max_{k\in N}\beta_k u_k\, \bigg|\, \beta_j u_j\geq \beta_i u_i\right] \\
    &< \mathbb{P}\left[\beta_j\geq \beta_i u_i\right]\cdot\mathbb{P}\left[\beta_j\geq\max_{k\in N,\ k\ne i, j}\beta_k u_k\right], 
\end{align*}
where the last inequality follows from $u_j\leq1$. Since
\begin{align*}
    \mathbb{P}\left[\beta_i u_i\leq \beta_j\right] = \mathbb{P}\left[u_i\leq \beta_j/\beta_i\right]\leq\mathbb{P}\left[u_i\leq1/(2q)\right]\leq 1/2, 
\end{align*}
we have
\begin{align*}
    \mathbb{P}\left[\beta_i u_i\geq\beta_j\right] \geq \mathbb{P}\left[\beta_j\geq \beta_i u_i\right]\quad \Rightarrow \quad p_i > p_j.
\end{align*}
\end{proof}
Here we extend \cref{lem:boundmult} to bound the ratios under approximated equalizing multipliers.
\begin{corollary}
If the set of multipliers $\vec{\beta}$ satisfies that $|p_i-1/n|<1/(2\, n)$ for all $i\in N$, then the ratio $\beta_{i}/\beta_{j}$ between any pair of agents $i, j$ is at most $4\, q$.
\label{cor:boundratio}
\end{corollary}

\begin{proof}
If $\beta_{i}/\beta_{j}>3\, q$ for some $i, j\in N$, then 
\begin{align*}
    \mathbb{P}\left[\beta_i u_i\leq \beta_j\right] = \mathbb{P}\left[u_i\leq \beta_j/\beta_i\right]\leq\mathbb{P}\left[u_i<1/(4q)\right]< 1/4, 
\end{align*}
Then following the inequalities in the above proof for \cref{lem:boundmult}, we would have $p_i>3\, p_j$, which contradicts to the fact that $1/(2\,n)<p_i, p_j<3/(2\,n)$.
\end{proof}

\begin{lemma}
\label{lem:smallstep}
Let $\vec{z} \in \mathbb{Z}^n$, $\epsilon > 0$, and $S \subseteq N$. For all $i \in S$,
\[ p_i\big((1+\epsilon)^{\vec{z}}\big) \leq p_i\big((1+\epsilon)^{\vec{z} + \bone_S}\big) \leq p_i\big((1+\epsilon)^{\vec{z}}\big) + 2 \, q \, \epsilon, \]
and, for all $i \notin S$,
\[ p_i\big((1+\epsilon)^{\vec{z}}\big) - 2 \, q \, \epsilon \leq p_i\big((1+\epsilon)^{\vec{z} + \bone_S}\big) \leq p_i\big((1+\epsilon)^{\vec{z}}\big). \]
\end{lemma}
\begin{proof}
This follows from repeated application of \cref{lem:localmonotonicity} and \cref{lem:lipschitzmultipliers}, as well as the observation that the effect on the resulting multipliers is equal whether we multiply all $\beta_i$ for $i \in S$ by $1 + \epsilon$ or whether we divide all $\beta_i$ for $i \notin S$ by $1 + \epsilon$ instead.
\end{proof}

\subsection{Uniqueness of Equalizing Multipliers}
\label{app:unique}
First, we formalize the notion of local strict monotonicity as the following lemma.
\begin{lemma}
Assuming the distributions for all agents have interval support. For any agent $j\in N$, let $p_j>0,\, p'_j$ be the probability that agent $j$ receives each item under $\{\beta_i\}_{i\in N}$ and $\{\beta'_i\}_{i\in N}$. If $\forall i\in N$, $\beta'_j/\beta'_i\geq \beta_j/\beta_i$ and there exists some $k\in N$ such that $p_k>0$ and $\beta'_j/\beta'_k> \beta_j/\beta_k$, then $p'_j>p_j$.
\label{lem:strict}
\end{lemma}
\begin{proof}
First we can express $p'_j, p_j$ as follows:
\begin{align*}
    &p'_j = \mathbb{P}\left[\beta'_j u'_j=\max_{i\in N}\{\beta'_i u'_i\}\right] = \int_{0}^{1}f_j(u) \prod_{i\in N/\{j\}}F_i\left(\frac{\beta'_j}{\beta'_i}u\right) du, \\
    &p_j = \int_{0}^{1}f_j(u) \prod_{i\in N/\{j\}}F_i\left(\frac{\beta_j}{\beta_i}u\right) du.
\end{align*}
Suppose $\mathcal{D}_j$'s support interval is $[\underline{u_j}, \overline{u_j}]$. Now take 
\begin{align*}
    \underline{u} = \max_u\{u: \exists i\in N,\, F_i\left(\frac{\beta_j}{\beta_i}u\right)=0\}.
\end{align*}
From $F_j\left(\frac{\beta_j}{\beta_j}\underline{u_j}\right)=0$, then by definition of $\underline{u}$, it is true that $\underline{u_j}\leq\underline{u}$.
We also have $\underline{u}<\overline{u_j}$, since otherwise we can find agent $i\in N$, such that $F_i\left(\frac{\beta_j}{\beta_i}u\right)=0$ for all possible value of $u_j=u\in[\underline{u_j}, \overline{u_j}]$, making $p_j=0$. Then we take 
\begin{align*}
    \overline{u} = \min_u\{u: F_k\left(\frac{\beta_j}{\beta_k}u\right)=1\}.
\end{align*}
We argue that $\overline{u}>\underline{u}$. Consider otherwise, then we can find $u_0\in [\overline{u}, \underline{u}]$ and $i\in N$, where
\begin{align*}
    &F_i\left(\frac{\beta_j}{\beta_i}u_0\right)=0,\, F_k\left(\frac{\beta_j}{\beta_k}u_0\right)=1 \\
    &\Rightarrow 
    u_i \geq \frac{\beta_j}{\beta_i}u_0,\,
    u_k \leq \frac{\beta_j}{\beta_k}u_0
    \Rightarrow \beta_i u_i\geq \beta_k u_k,
\end{align*}
which will make $p_k=0$ since all distributions are non-atomic, contradicting our assumption that $p_k>0$. Combining the earlier arguments, we know that $\max\{\underline{u_j}, \underline{u}\} < \min\{\overline{u_j}, \overline{u}\}$. Then for any $u\in  (\max\{\underline{u_j}, \underline{u}\}, \min\{\overline{u_j}, \overline{u}\})=\mathcal{I}$,
\begin{align*}
    &f_j(u)>0, \quad F_i\left(\frac{\beta_j}{\beta_i}u\right) > 0,\, \forall i\in N, \\
    &0<F_k\left(\frac{\beta_j}{\beta_k}u\right) < 1
    \Rightarrow F_k\left(\frac{\beta'_j}{\beta'_k}u\right) > F_k\left(\frac{\beta_j}{\beta_k}u\right).
\end{align*}
The last inequality follows from the fact that the derivative $F'_k(u)>0$ for $u$ that satisfies $0<F_k(u)<1$ (guaranteed by Interval support property, $0<F_k(u)<1$ just means $u$ is in the support interval), and that $\beta'_j/\beta'_k>\beta_j/\beta_k$. Then 
\begin{align*}
    \int_{u\in\mathcal{I}}f_j(u) \prod_{i\in N/\{j\}}F_i\left(\frac{\beta'_j}{\beta'_i}u\right) du 
    > \int_{u\in\mathcal{I}}f_j(u) \prod_{i\in N/\{j\}}F_i\left(\frac{\beta_j}{\beta_i}u\right) du.
\end{align*}
For $u$ in the rest of the range, we have that
\begin{align*}
\frac{\beta'_j}{\beta'_i}\geq \frac{\beta_j}{\beta_i}
\Rightarrow
F_i\left(\frac{\beta'_j}{\beta'_i}u\right)\geq F_i\left(\frac{\beta_j}{\beta_i}u\right),\, \forall i\in N/\{j\}.
\end{align*}
Then the integral in this range for $p'_j$ is greater or equal to that for $p_j$. Hence we have showed the strict ordering $p'_j>p_j$.
\end{proof}
Now we prove uniqueness of equalizing multipliers. Suppose there are two different sets of equalizing multipliers, namely $\vec{\beta}, \vec{\beta'}$ (assume they are both normalized by setting $\beta_1=1$).
We can find the $i=\arg\max_{j\in N}\beta'_j/\beta_j$, w.l.o.g. assume $\beta'_{i}/\beta_{i}>1$, otherwise we simply exchange $\vec{\beta}$ and $\vec{\beta'}$ and the maximum ratio must be larger than $1$ since the ratios cannot all be 1.
Now we know that for any $j\in N$, $\beta'_i/\beta'_j = (\beta'_i/\beta_i)/(\beta'_j/\beta_j)\cdot\beta_i/\beta_j \geq \beta_i/\beta_j$, while $\beta'_i/\beta'_1>\beta_i/\beta_1$. Then by the local strict monotonicity \cref{lem:strict}, $i$'s probability under $\vec{\beta'}$ is strictly larger than that under $\vec{\beta}$, contradicting to them both being $1/n$ under equalizing multipliers.

\subsection{Inequalities for Expectations}
\label{app:exp}
\begin{lemma}
For any pair of agents $i, j\in N$, the following inequalities between expectations hold (assuming that the conditions in the conditional expectations can be met):
\begin{equation}
    \mathbb{E}\left[u_i\, \bigg|\, \beta_j u_j = \max_{k\in N}\beta_k u_k\right] \leq \mathbb{E}\left[u_i\right]
\label{eq:exp1}
\end{equation}
and
\begin{equation}
\mathbb{E}\left[u_i\, \bigg|\, \beta_i u_i=\max_{k\in N}\beta_k u_k\right] \geq \mathbb{E}\left[u_i\, \big|\, \beta_i u_i \geq \beta_j u_j\right].
\label{eq:exp2}
\end{equation}
\end{lemma}
\begin{proof}
Let $v_j=\mathbbm{1} \left[\beta_j u_j = \max_{k\in N}\beta_k u_k\right]$, which is a random variable taking value from $\{0, 1\}$.
Then we have
\begin{align*}
    \mathbb{E}\left[u_i\, \bigg|\, \beta_j u_j = \max_{k\in N}\beta_k u_k\right]
    = \mathbb{E}\left[u_i\, \big|\, v_j=1\right]
    = \int_{0}^{1}f_{u_j|v_j=1}(u)\cdot \mathbb{E}\left[u_i\, \big|\, v_j=1, u_j=u\right] du.
\end{align*}
Given $v_j=1$, for all $u\in[0, 1]$, it holds that
\begin{align*}
    \mathbb{E}\left[u_i\, \big|\, v_j=1, u_j=u\right]
    = \mathbb{E}\left[u_i\, \big|\, \beta_i u_i\leq \beta_j u\right]
    \leq \mathbb{E}\left[u_i\right].
\end{align*}
Thus
\begin{align*}
    \int_{0}^{1}f_{u_j|v_j=1}(u)\cdot \mathbb{E}\left[u_i\, \big|\, v_j=1, u_j=u\right] du\leq 
    \int_{0}^{1}f_{u_j|v_j=1}(u)\cdot \mathbb{E}\left[u_i\right] du
    = \mathbb{E}\left[u_i\right].
\end{align*}
Therefore \cref{eq:exp1} holds.

For \cref{eq:exp2}, let $v_{i,j}=\max_{k\in N/\{i, j\}}\beta_k u_k$ be a random variable taking value in $[0, 1]$, and $v_i=\mathbbm{1} \left[\beta_i u_i = \max_{k\in N}\beta_k u_k\right]$.
Then by substituting the condition we have
\begin{align*}
    &\mathbb{E}\left[u_i\, \bigg|\, \beta_i u_i=\max_{k\in N}\beta_k u_k\right]
    = \mathbb{E}\left[ u_i\, \bigg|\, v_i=1\right] 
    = \mathbb{E}\left[ u_i\, \bigg|\, \beta_i u_i\geq \beta_j u_j, \beta_i u_i\geq v_{i,j}\right] \\
    &= \int_{0}^{1}f_{v_{i,j}|v_i=1}(u)\cdot\mathbb{E}\left[ u_i\, \bigg|\, \beta_i u_i\geq \beta_j u_j, \beta_i u_i\geq v_{i,j}, v_{i,j}=u\right]du.
\end{align*}
Since $v_{i,j}, u_i, u_j$ are independent, for all $u\in[0,1]$ it holds that
\begin{align*}
    \mathbb{E}\left[ u_i\, \bigg|\, \beta_i u_i\geq \beta_j u_j, \beta_i u_i\geq v_{i,j}, v_{i,j}=u\right]
    = \mathbb{E}\left[ u_i\, \bigg|\, \beta_i u_i\geq \beta_j u_j, \beta_i u_i\geq u\right]
    \geq \mathbb{E}\left[ u_i\, \bigg|\, \beta_i u_i\geq \beta_j u_j\right].
\end{align*}
Hence
\begin{align*}
    &\int_{0}^{1}f_{v_{i,j}|v_i=1}(u)\cdot\mathbb{E}\left[ u_i\, \bigg|\, \beta_i u_i\geq \beta_j u_j, \beta_i u_i\geq v_{i,j}, v_{i,j}=u\right]du\\
    &\geq \int_{0}^{1}f_{v_{i,j}|v_i=1}(u)\cdot\mathbb{E}\left[ u_i\, \bigg|\, \beta_i u_i\geq \beta_j u_j\right]du
    = \mathbb{E}\left[ u_i\, \bigg|\, \beta_i u_i\geq \beta_j u_j\right].
\end{align*}
Therefore \cref{eq:exp2} holds.
\end{proof}

\subsection{Positive Gap Between Conditional Expectations}
\label{app:positivegap}
\begin{restatable}{proposition}{proppositivegap}
\label{prop:positivegap}
Fix a set of agents whose utility distributions have interval support, and let $\{\beta_i\}_{i \in N}$ denote a set of multipliers such that $\left|p_i-1/n\right|< 1/(2\,n)$ for all $i\in N$.
Then, for all $i \neq j$, \[\mathbb{E}\left[u_i \middle| \beta_i u_i = \max_{k \in N} \beta_k u_k \right] > \mathbb{E}\left[u_i \middle| \beta_j u_j = \max_{k \in N} \beta_k u_k \right].\]
\end{restatable}
\begin{proof}
Let $X \coloneqq \beta_i u_i$ and $Y \coloneqq \beta_j u_j$ denote the scaled random variables, then it still holds that: (a) $X$ and $Y$ are independent, (b) $X$ and $Y$ have interval support, and (c) both $\mathbb{P}\left[X > Y\right]$ and $\mathbb{P}\left[Y > X\right]$ are at least $1/(2\,n)>0$.
From \cref{app:exp}, we know that
\begin{align*}
    &\mathbb{E}\left[u_i\, \bigg|\, \beta_i u_i=\max_{k\in N}\beta_k u_k\right] \geq \mathbb{E}\left[u_i\, |\, \beta_i u_i\geq\beta_j u_j\right]
    = \frac{1}{\beta_i}\mathbb{E}\left[X\, \bigg|\, X> Y\right], \\
    &\mathbb{E}\left[u_i\, \bigg|\, \beta_j u_j=\max_{k\in N}\beta_k u_k\right]  \leq \mathbb{E}\left[u_i\right]
    = \frac{1}{\beta_i}\mathbb{E}\left[X\right].
\end{align*}
Then it suffices to show that $\mathbb{E}\left[X \mid X > Y \right] > \mathbb{E}\left[X\right]$.

Let $\mathcal{I}$ denote the intersection of the support intervals of $X$ and $Y$, excluding both endpoints.
This intersection is a nonempty interval and both $\mathbb{P}\left[X \in \mathcal{I}\right]$ and $\mathbb{P}\left[Y \in \mathcal{I}\right]$ must have positive probability, since, else, one variable's support would entirely lie below or above the other variable's support, which would contradict the above observation that $0 < \mathbb{P}\left[X>Y\right] < 1$.
Since for all $y\in \mathcal{I}$, 
\begin{align*}
    \mathbb{E}\left[X\right] = 
    \underbrace{\mathbb{P}\left[X>y\right]}_{>0} \cdot \underbrace{\mathbb{E}\left[X\, |\, X > y\right]}_{>y} + \underbrace{\mathbb{P}\left[X\leq y\right]}_{=1-\mathbb{P}\left[X>y\right] > 0} \cdot \underbrace{\mathbb{E}\left[X\, |\, X \leq y\right]}_{\leq y}
    \, \Rightarrow\,  \mathbb{E}\left[X\, |\, X > y\right] > \mathbb{E}\left[X\right].
\end{align*}
Hence
\begin{align*}
    \mathbb{E}\left[X\, |\, X > Y, Y\in \mathcal{I}\right]
    &= \int_{0}^{\beta_j}f_{Y|Y<X, Y\in\mathcal{I}}(y)\cdot\mathbb{E}\left[X\, |\, X > Y, Y\in \mathcal{I}, Y=y\right] dy \\
    &= \int_{0}^{\beta_j}f_{Y|Y<X, Y\in\mathcal{I}}(y)\cdot\mathbb{E}\left[X\, |\, X > y\right] dy \\
    &> \int_{0}^{\beta_j}f_{Y|Y<X, Y\in\mathcal{I}}(y)\cdot\mathbb{E}\left[X\right] dy 
    = \mathbb{E}\left[X\right].
\end{align*}
For all $y\notin \mathcal{I}$, $\mathbb{E}\left[X\, |\, X > y\right] \geq \mathbb{E}\left[X\right]$, from which we have
\begin{align*}
    \mathbb{E}\left[X\, |\, X > Y, Y\notin \mathcal{I}\right]
    &= \int_{0}^{\beta_j}f_{Y|Y<X, Y\notin\mathcal{I}}(y)\cdot\mathbb{E}\left[X\, |\, X > Y, Y\notin \mathcal{I}, Y=y\right] dy \\
    &= \int_{0}^{\beta_j}f_{Y|Y<X, Y\notin\mathcal{I}}(y)\cdot\mathbb{E}\left[X\, |\, X > y\right] dy \\
    &\geq \int_{0}^{\beta_j}f_{Y|Y<X, Y\notin\mathcal{I}}(y)\cdot\mathbb{E}\left[X\right] dy 
    = \mathbb{E}\left[X\right].
\end{align*}
Then, we can bound
\begin{align*}
    \mathbb{E}\left[X\, |\, X > Y\right] 
    ={} &\underbrace{\mathbb{P}\left[Y\in \mathcal{I} \mid X > Y \right]}_{>0} \cdot \underbrace{\mathbb{E}\left[X\, |\, X > Y, Y\in \mathcal{I}\right]}_{> \mathbb{E}[X]} \\
    &+ \underbrace{\mathbb{P}\left[Y \notin \mathcal{I} \mid X > Y \right]}_{= 1 - \mathbb{P}\left[Y\in \mathcal{I} \mid X > Y \right]} \cdot \underbrace{\mathbb{E}\left[X\, |\, X > Y, Y\notin \mathcal{I}\right]}_{\geq \mathbb{E}[X]} \\
    >{} &\mathbb{E}\left[X\right],
\end{align*}
which shows a positive gap.

\end{proof}

\subsection{Discussion of $(p,q)$-PDF-Boundedness}
\label{app:pdfboundeddisc}
The above proof in \cref{app:positivegap} shows a gap for any particular set of agents, but, to bound the probability of EF when the number and distributions of agents change, we need a uniform constant lower bound for all $n$ and all utility distributions involved.
Obtaining such a bound requires some additional restriction on which distributions are allowed, and $(p, q)$-PDF-boundedness is a natural choice for this:
On the one hand, excessively low probability densities cn make the gap grow arbitrarily small, although still $>0$. To see this, consider a variant of the counter-example where agent~A’s distribution is uniform on $[1/4, 3/4]$ and agent~B’s utility is uniformly distributed on $[0, 1/4]\cup[3/4, 1]$, in which we increase the density of agent B's distribution in $[1/4, 3/4]$ to a small, positive constant $\epsilon$. It is easy to see that, in this example, allowing arbitrarily low (positive) densities can cause the gap to become arbitrarily small.
On the other hand, the gap could vanish as a result of excessively high rather than low densities. Indeed, in a scenario where agent~A's utility is uniform on $[0,1]$ and agent B's distribution is uniform on $[1/2-\epsilon, 1/2 + \epsilon]$, as $\epsilon \to 0^+$ and agent B's density grows unboundedly, the positive gap for this agent goes to zero.
Assuming that all densities (in the support) lie between some constants $p>0$ and $q$ avoids these problematic cases.

\subsection{Proof of \cref{lem:boundgap}}
Before proving the constant gap, we first show a lower bound on the length of the intersection between support intervals.
\begin{lemma}
Suppose the set of multipliers $\vec{\beta}$ satisfies that $|p_i-1/n|<1/(2\, n)$ for all $i\in N$.
For any $i, j\in N$ and $h(u)=f_j\left(\frac{\beta_i}{\beta_j}u\right)$, the interval $I^* = \text{supp}(h)\cap\text{supp}(f_i)$ has length at least $L_{q}$, which only depends on $q$.
\label{lem:overlap}
\end{lemma}
\begin{proof}
We consider the random variables $u_i\sim D_i,\ u_j\sim D_j$ after scaled: $\widetilde{u_i}=\beta_i u_i,\ \widetilde{u_j}=\beta_j u_j$, which have PDF: $\widetilde{f_i}(u)=(1/\beta_i)\,f_i\left(u/\beta_i\right)$ and $\widetilde{f_j}(u)=(1/\beta_j)\,f_j\left(u/\beta_j\right)$.
Without loss of generality we will assume that the $\min_{k\in N}\beta_k=1$, then by \cref{cor:boundratio} we know that $\max_{k\in N}\beta_k\leq 4\,q$.

We will first prove that $\text{supp}(\widetilde{f_i})\cap \text{supp}(\widetilde{f_j})$ has lower bounded length. 
Note that $\widetilde{f_i}$ is $(p/\beta_i, q/\beta_i)$-PDF-bounded and $\widetilde{f_j}$ is $(p/\beta_j, q/\beta_j)$-PDF-bounded.

If one of the two support intervals contain the other, then the length of their intersection will be at least $\min\{\beta_i/q, \beta_j/q\}\geq 1/q$.

Otherwise, if $\text{supp}(\widetilde{f_i})$ lies on the right of $\text{supp}(\widetilde{f_j})$, suppose their intersection is $[a, b]$ (there will not be a vacant intersection since then $\beta_ju_j$ would always be smaller than $\beta_iu_i$), then it always holds that $\beta_iu_i\geq a,\, \beta_ju_j\leq b$. We claim that the length of $[a, b]$ is at least $1/(4\,q)$. Otherwise, consider 
\begin{align*}
    p_j &= \mathbb{P}\left[\beta_ju_j=\max_{k\in N}\beta_ku_k\right] \\
    &= \mathbb{P}\left[\beta_iu_i\leq b\right]\cdot\mathbb{P}\left[\beta_ju_j=\max_{k\in N}\beta_ku_k\, \bigg|\, \beta_iu_i\leq b\right] \\
    &\leq \mathbb{P}\left[\beta_iu_i\leq b\right]\cdot\mathbb{P}\left[\max_{k\in N, k\ne i}\beta_ku_k\leq b\right],
\end{align*}
while
\begin{align*}
    p_i &= \mathbb{P}\left[\beta_iu_i=\max_{k\in N}\beta_ku_k\right] \\
    &\geq \mathbb{P}\left[\beta_iu_i\geq b\right]\cdot\mathbb{P}\left[\beta_iu_i=\max_{k\in N}\beta_ku_k\, \bigg|\, \beta_ju_j\geq b\right] \\
    &\geq \mathbb{P}\left[\beta_iu_i\geq b\right]\cdot\mathbb{P}\left[\max_{k\in N, k\ne i}\beta_ku_k\leq b\right].
\end{align*}
If the length of $[a, b]$ is less than $1/(4\,q)$, we have
\begin{align*}
    \mathbb{P}\left[\beta_iu_i\leq b\right] = \mathbb{P}\left[\beta_iu_i\in [a, b]\right] < \frac{q}{\beta_i}\cdot\frac{1}{4q} \leq \frac{1}{4}, 
\end{align*}
which indicates that $p_i>3\, p_j$ while it should be true that $1/(2\,n)<p_i, p_j<3/(2\,n)$, hence the contradiction.
The argument is symmetric for the case where $\text{supp}(\widetilde{f_i})$ lies on the left of $\text{supp}(\widetilde{f_j})$, which also gives the same lower bound $1/(4\,q)$ on the length.

Therefore we conclude that $I = [a, b] = \text{supp}(\widetilde{f_i})\cap \text{supp}(\widetilde{f_j})$ has length at least $1/(4\,q)$.
Then $f_j$ is supported on $[a/\beta_j, b/\beta_j]$ and $f_i$ is supported on $[a/\beta_i, b/\beta_i]$. Thus
\begin{align*}
\forall u\in[\frac{a}{\beta_i}, \frac{b}{\beta_i}],\, \frac{\beta_i}{\beta_j}u\in [\frac{a}{\beta_j}, \frac{b}{\beta_j}],\, h(u)=f_j(\frac{\beta_i}{\beta_j}u)>0 \\
\Rightarrow [a/\beta_i, b/\beta_i]\subseteq\text{supp}(h),
\end{align*}
and
\begin{align*}
    \,\,\,\,\forall u\in[\frac{a}{\beta_i}, \frac{b}{\beta_i}],\ f_i(u)>0\, \Rightarrow\, [a/\beta_i, b/\beta_i]\subseteq\text{supp}(f_i).
\end{align*}
Therefore the interval $[a/\beta_i, b/\beta_j]\subseteq \text{supp}(h)\cap\text{supp}(f_i)$, with a length of at least $1/(4\,q)\cdot1/(4\,q)=1/(16\,q^2)$. This proves our lemma that the intersection $I^*$ has length at least $L_{q}=1/(16\,q^2)$.

\end{proof}

\label{app:boundgap}
\lemboundgap*
\begin{proof}
In \cref{app:exp} we show that
\begin{align*}
    \mathbb{E}\left[u_i\, \bigg|\, \beta_j u_j = \max_{k\in N}\beta_k u_k\right] \leq \mathbb{E}\left[u_i\right]
\end{align*}
and
\begin{align*}
\mathbb{E}\left[u_i\, \bigg|\, \beta_i u_i=\max_{k\in N}\beta_k u_k\right] \geq \mathbb{E}\left[u_i\, \big|\, \beta_i u_i \geq \beta_j u_j\right].
\end{align*}
Then it suffices to show that there is a constant gap between $\mathbb{E}\left[u_i\, |\, \beta_i u_i > \beta_j u_j\right]$ and $\mathbb{E}[u_i]$. Let
\begin{align*}
\mathcal{P} = \mathbb{P}\left[\beta_i u_i \geq \beta_j u_j\right]=\int_{0}^{1}f_i(u)\, F_j\left(\frac{\beta_i}{\beta_j}u\right) du,
\end{align*}
and 
\begin{align*}
\Delta\mathbb{E} = \mathbb{E}\left[u_i\, \big|\, \beta_i u_i \geq \beta_j u_j\right]-\mathbb{E}[u_i],
\end{align*}
then we have
\begin{align*}
    \Delta\mathbb{E} &= \frac{1}{\mathcal{P}}\int_{0}^{1}u\, f_i(u) \left(F_j\left(\frac{\beta_i}{\beta_j}u\right)-\mathcal{P}\right) du.
\end{align*}
Let $g(u)=F_j\left(\frac{\beta_i}{\beta_j}u\right)-\mathcal{P}$, and $h(u)=f_j\left(\frac{\beta_i}{\beta_j}u\right)$.
It is clear that $g(u)$ is monotonically increasing, moreover, we can lower bound the derivative of $g(u)$ on the support of $h(u)$, which is an interval and we denote this range as $\text{supp}(h)$:
\begin{align*}
g'(u) = \frac{\beta_i}{\beta_j}\, f_j\left(\frac{\beta_i}{\beta_j}u\right) \geq \frac{p}{4q},\ u\in \text{supp}(h).
\end{align*}
The inequality holds since $\beta_j/\beta_i\leq4\,q$ from \cref{cor:boundratio}.
Let $\text{supp}(f_i)$ denote the support interval of $f_i(u)$. 
In \cref{lem:overlap} we show a lower bound, $L_{q}$, which only depends on $q$, on the length of the interval $I^* = \text{supp}(h) \cap \text{supp}(f_i)$. 
We consider such interval $I^*=[l^*, r^*]\subseteq [0, 1]$ with midpoint $m^*=(l^* + r^*)/2$, and we know that $r^*-l^*\geq L_{q}$.
From the previous analysis, we know that for any $u\in I^*$, it always holds that $f_i(u)\geq p$ and $g'(u)\geq D_{p, q}=p/(4\,q)$.

Since
\begin{equation}
    \int_{0}^{1}f_i(u)\, g(u)\, du=0,
\label{eq:int}
\end{equation}
combined with $g(u)$'s continuity and monotonicity, there exists a point $u^*\in[0, 1]$ where $g(u^*)=0$. The interval $I^*$ must have at least half of its length that lies on the left or right side of $u^*$, without loss of generality we assume that $u^*\leq m^*$ and interval $[m^*, r^*]$ lies on the right of $u^*$.
Let
\begin{align*}
    c_1 = \int_{u^*}^{\frac{m^*+r^*}{2}}f_i(u)\, g(u)\, du,\ 
    c_2 = \int_{\frac{m^*+r^*}{2}}^{1}f_i(u)\, g(u)\, du.
\end{align*}
When $u\in[u^*, \frac{m^*+r^*}{2}]$, we have $f_i(u)\geq0,\, g(u)\geq g(u^*)=0$, hence $c_1\geq0$. While $u\in[\frac{m^*+r^*}{2}, r^*]\subseteq I^*$, we have that $f_i(u)\geq p$ and
\begin{align*}
    g(u)&\geq g\left(\frac{m^*+r^*}{2}\right) \geq g(m^*)+\frac{r^*-m^*}{2}\cdot D_{p, q} \\
    &\geq g(u^*)+\frac{L_{p}}{4}\cdot D_{p, q} = \frac{L_{p}\, D_{p, q}}{4}.
\end{align*}
Then we can lower bound $c_2$ by a positive constant $G_{p, q}$:
\begin{align*}
    c_2 \geq \int_{\frac{m^*+r^*}{2}}^{r^*}f_i(u)\, g(u)\, du \geq \frac{r^*-m^*}{2}\cdot p\cdot \frac{L_{q}\, D_{p, q}}{4} 
    \geq \frac{p\, L^2_{q}\, D_{p, q}}{16} = G_{p,q}.
\end{align*}
\cref{eq:int} indicates that
\begin{align*}
    \int_{0}^{u^*}f_i(u)\, g(u)\, du = -(c_1+c_2).
\end{align*}
Then we have
\begin{align*}
    &\int_{0}^{u^*}u\, f_i(u)\, g(u)\, du \geq -u^*\, (c_1+c_2), \\
    &\int_{u^*}^{\frac{m^*+r^*}{2}}u\, f_i(u)\, g(u)\, du \geq u^*\, c_1, \\
    &\int_{\frac{m^*+r^*}{2}}^{1}u\, f_i(u)\, g(u)\, du \geq \frac{m^*+r^*}{2}\, c_2,
\end{align*}
and we can lower bound $\Delta\mathbb{E}$ by
\begin{align*}
    \Delta\mathbb{E} &\geq -u^*(c_1+c_2)+u^*c_1+\frac{m^*+r^*}{2}c_2 \\
    &= \left(\frac{m^*+r^*}{2}-u^*\right)c_2 \geq \frac{L_{q}\,G_{p, q}}{4}.
\end{align*}
Therefore we have finished the proof with 
\begin{align*}
    C_{p, q} = \frac{L_{p}\,G_{p, q}}{4} = \frac{p^2}{2^{20} \,q^7}\leq1.
\end{align*}
\end{proof}

\subsection{Envy-free: Combining Previous Results}
\label{app:combine}
For any two agents $i, j\in N$, and each item $\alpha\in M$, let $X_\alpha$ denote its contribution to $u_i(A_i)$ and $Y_\alpha$ denote its contribution to $u_i(A_j)$.
In particular, $X_\alpha=0$ if $\alpha\notin A_i$ and $X_\alpha=u_i(\alpha)$ if $\alpha\in A_i$; while $Y_\alpha=0$ if $\alpha\notin A_j$ and $Y_\alpha=u_i(\alpha)$ if $\alpha\in A_j$.
When the set of equalizing multipliers are approximated such that $|p_i-1/n|\leq\delta=\min(C_{p, q}/(4\,n), 1/(4\,n))$ for all $i\in N$, following from the deduced constant gap in \cref{lem:boundgap}, we have
\begin{align*}
    \mathbb{E}\left[X_\alpha\right] &= \mathbb{P}\left[\alpha\in A_i\right] \cdot \mathbb{E}\left[u_i(\alpha)\, \big|\, \alpha\in A_i\right] \\
    &= \mathbb{P}\left[\beta_i u_i = \max_{k\in N}\beta_k u_k\right]\cdot \mathbb{E}\left[u_i\, \bigg|\, \beta_i u_i = \max_{k\in N}\beta_k u_k\right] \\
    &\geq (\frac{1}{n}-\delta)\cdot\mathbb{E}\left[u_i\, \bigg|\, \beta_i u_i = \max_{k\in N}\beta_k u_k\right], \\
    \mathbb{E}\left[Y_\alpha\right] &= \mathbb{P}\left[\alpha\in A_j\right] \cdot \mathbb{E}\left[u_i(\alpha)\, \big|\, \alpha\in A_j\right] \\
    &= \mathbb{P}\left[\beta_j u_j = \max_{k\in N}\beta_k u_k\right]\cdot \mathbb{E}\left[u_i\, \bigg|\, \beta_j u_j = \max_{k\in N}\beta_k u_k\right] \\
    &\leq (\frac{1}{n}+\delta)\cdot\mathbb{E}\left[u_i\, \bigg|\, \beta_j u_j = \max_{k\in N}\beta_k u_k\right], \\
    &\quad\Rightarrow \mathbb{E}\left[X_\alpha\right] - \mathbb{E}\left[Y_\alpha\right] \geq C_{p, q}/n-2\delta\geq C_{p, q}/(2n).
\end{align*}
Note that $X_\alpha, Y_\alpha$ are independently and identically distributed for all $\alpha\in M$, and
\begin{align*}
    u_i(A_i) = \sum_{\alpha\in M}X_\alpha, \quad
    u_i(A_j) = \sum_{\alpha\in M}Y_\alpha.
\end{align*}
Thus we can bound $u_i(A_i)$ by Chernoff's bound: for any $\alpha$, 
\begin{align*}
    &\mathbb{P}\left[u_i(A_i) < \left(1-\frac{C_{p, q}}{4\, n\, \mathbb{E}[X_\alpha]}\right) m\, \mathbb{E}\left[X_\alpha\right]\right] \\
    &\leq \exp\left( -\frac{m\, C^2_{p, q}}{32\, n^2\, \mathbb{E}\left[X_\alpha\right]}\right)
    \leq \exp\left(-\frac{m\, C^2_{p, q}}{32\, n}\right).
\end{align*}
where the last inequality follows from $\mathbb{E}\left[X_\alpha\right]\leq 1/n$. Similarly we can bound $u_i(A_j)$:
\begin{align*}
    &\mathbb{P}\left[u_i(A_j) > \left(1+\frac{C_{p, q}}{4\, n\, \mathbb{E}\left[Y_\alpha\right]}\right) m\, \mathbb{E}\left[Y_\alpha\right]\right] \\
    &\leq \exp\left( -\frac{m\, C^2_{p, q}}{48\, n^2\, \mathbb{E}\left[Y_\alpha\right]}\right)
    \leq \exp\left(-\frac{m\, C^2_{p, q}}{48\, n}\right).
\end{align*}
Then we can use union bound to bound the probability $\mathcal{P}_{ij}$ that neither of the above two events happen. With probability 
\begin{align*}
    &\mathcal{P}_{ij} \geq 1 - \exp\left(-\frac{m\, C^2_{p, q}}{32\, n}\right) - \exp\left(-\frac{m\, C^2_{p, q}}{48\, n}\right) \\
    &\geq 1-\frac{2}{n^2}\exp\left(2\, \log n-\frac{m\, C^2_{p, q}}{32\, n}\right),
\end{align*}
we have $u_i(A_i)\geq u_i(A_j)$.
Again we use union bound on the probability that for any $i, j\in N, u_i(A_i)\geq u_i(A_j)$ (it just means that the allocation is envy-free):
\begin{align*}
    \mathcal{P} \geq 1-2\, \exp\left(2\, \log n-\frac{m\, C^2_{p, q}}{32\, n}\right),
\end{align*}
which suggests that the allocation is envy-free with probability at least $1-2\,\exp\left(2\, \log n-\mathit{const}(p, q)\, m/n\right)$.

\section{Negative Result for Round Robin Algorithm}
\label{app:neg}
Suppose agent 1 has uniform distribution on $[0.6, 1]$, and agent $n$ has uniform distribution on $[0, 1]$. Assume in the round robin algorithm, there are a total of $t$ rounds. Agent 1 gets the first item, and agent $n$ gets the last item. With probability $p_1=1/3$, event A happens, where agent $n$ values the item that agent 1 gets in the first round at least $2/3$.

Consider the last two items that agent $n$ gets, let $X_{t-1}, X_{t}$ denote agent $n$'s utility on them. From Lemma 3.2 of \citet{MS20}, we know their distribution is $X_{t-1}\sim\mathcal{D}^{\max(n+1)}_{\leq X_{t-2}}, X_t\sim \mathcal{D}^{\max(1)}_{\leq X_{t-1}}$, where $\mathcal{D}^{\max(k)}_{\leq T}$ denote the distribution of the maximum of $k$ samples drawn from $\mathcal{D}$ truncated at $T$. $\mathcal{D}^{\max(n+1)}_{\leq X_{t-2}}$ is stochastically dominated by $\mathcal{D}^{\max(n+1)}$, which is the $(n+1)$th order statistic of $(n+1)$ samples. Hence, with probability at least $p_2=(1/3)^{n+1}$, event B happens, where $X_{t-1}\leq 1/3$, and since $X_t\leq X_{t-1}$ we also have $X_t\leq 1/3$. 

The probability that both events happen is at least $p_1\cdot p_2=(1/3)^{n+2}$, since event A and event B are independent. 
When both event A and event B happen, consider agent $1$ trading the first item she gets for the last two items agent $n$ gets, then agent $1$'s utility strictly increase, since $0.6+0.6>1$; at the same time, agent $n$'s utility does not decrease, since $X_{t-1}+X_{t}\leq 2/3$. Then the original allocation is Pareto dominated by the allocation after this trade. Therefore, when $n=\Theta(1)$, this means that with constant probability, the allocation with round robin algorithm is not Pareto-optimal, no matter how large $m$ is.
As the dash-dotted line in \cref{fig:experiments} suggests, for larger $m$, such trade for Pareto improvement is more prevalent.

\section{Details on Empirical Results}
\label{app:experiment1}

\subsection{Setup}
Code for all our experiments can be found at
\url{https://github.com/pgoelz/asymmetric}.
We implemented the approximate multiplier algorithm and the main experiments in Python (3.7.10).
We rely on Scipy for evaluating integrals (\texttt{integrate.quad}) and for the PDF and CDF of the beta distributions.
We optimize fractional Maximum Nash Welfare using cvxpy (1.1.14), which in turn calls the MOSEK solver (9.2.9).
Integer MNW allocations are found using the Baron (2020.4.14), called through pyomo (6.1.2).
To check Pareto-optimality, we use the Gurobi solver (9.0.3).
Finally, we use Numpy in version 1.17.3.

All experiments were run on a MacBook Pro with a 3.1 GHz Dual-Core i5 processor and 16 GB RAM, running macOS 10.15.7.
To verify the allocation probabilities, and for \cref{fig:percentile}, we use Mathematica 13.0.0 on the same machine (Mathematica code is included in the above Git repository).

\subsection{Utility Distributions}
The ten utility distributions that we study in the body of the paper all come from a parametric family, which we will call peak distributions.
Specifically, each peak distribution $\mathcal{D}^\mathit{peak}_a$ is parameterized by its \emph{peak} $a \in (0, 1)$, and has the PDF $f_a$ where
\[ f_a(x) = \begin{cases}
1/10 + (18/10) \, x / a & \text{if $x \leq a$} \\
1/10 + (18/10) \, (1 - x) / (1 - a) & \text{else}.
\end{cases}\]
This PDF grows linearly from $f_a(0) = 1/10$ to $f_a(a) = 19/10$ and decreases in another linear segment from there to $f_a(1) = 1/10$.
Thus, the distributions are all $(1/10, 19/10)$-PDF-bounded as claimed, and have different means and skews depending on their peak.
We study the distributions $\{\mathcal{D}^\mathit{peak}_{i/11}\}_{i=1, \dots, 10}$.

\subsection{Computing Maximum Nash Welfare}
As we mention in the body, finding discrete MNW allocations with BARON ran too slowly for our experiments.
For example, allocating $50$ random items took on average 5 minutes each time (over 50 random instances); since we estimate probabilities by sampling 1\,000 random instances, this datapoint (with still modest $m$) would take about three days to compute.
Given that finding a MNW allocation is NP-complete, this failure is to be expected.

The main trick for computing MNW in the literature, due to \citet{CKM+19}, assumes that the utilities of agents can only take on a small set of integer values.
Unfortunately, this trick is not applicable in our setting, since we have many more items and since discretizing the utilities is likely to introduce violations of PO.

Instead, we propose to optimize a convex program to find the \emph{fractional} allocation with maximal Nash welfare, which is much more efficient.
Then, we round the fractional allocation into a proper allocation, by assigning each item to the agent who receives
the largest share of the item in the fractional allocation.
Performing any such rounding on an fPO fractional allocation preserves fractional Pareto optimality, which is easy to see given the characterization of fPO given in our Model section.

\subsection{Divisibility by $n$}
\label{app:experiment2}
In the experiment displayed in \cref{fig:experiments} in the body of the paper, we evaluate the following sequence of items: $m \in \{10, 20, 100, 200, 500, 1\,000, 2\,000, 5\,000, 10\,000\}$.
This progression is a natural way to explore the space of $m$ on a logarithmic axis, but comes with one caveat: All $m$ are divisible by $n=10$, a special case in which envy-free allocations are known to appear at lower $m$ than in the general case \citep{MS19}.
To verify that our empirical findings are robust to $m$ that are not multiples of $n$, we repeat the experiment for the first values of $m$, but shifting each $m$ by $3$ as follows: $m \in \{13, 23, 53, 103, 203, 503, 1003\}$.
\begin{figure}[h]
    \centering
    \includegraphics[width=.5\linewidth]{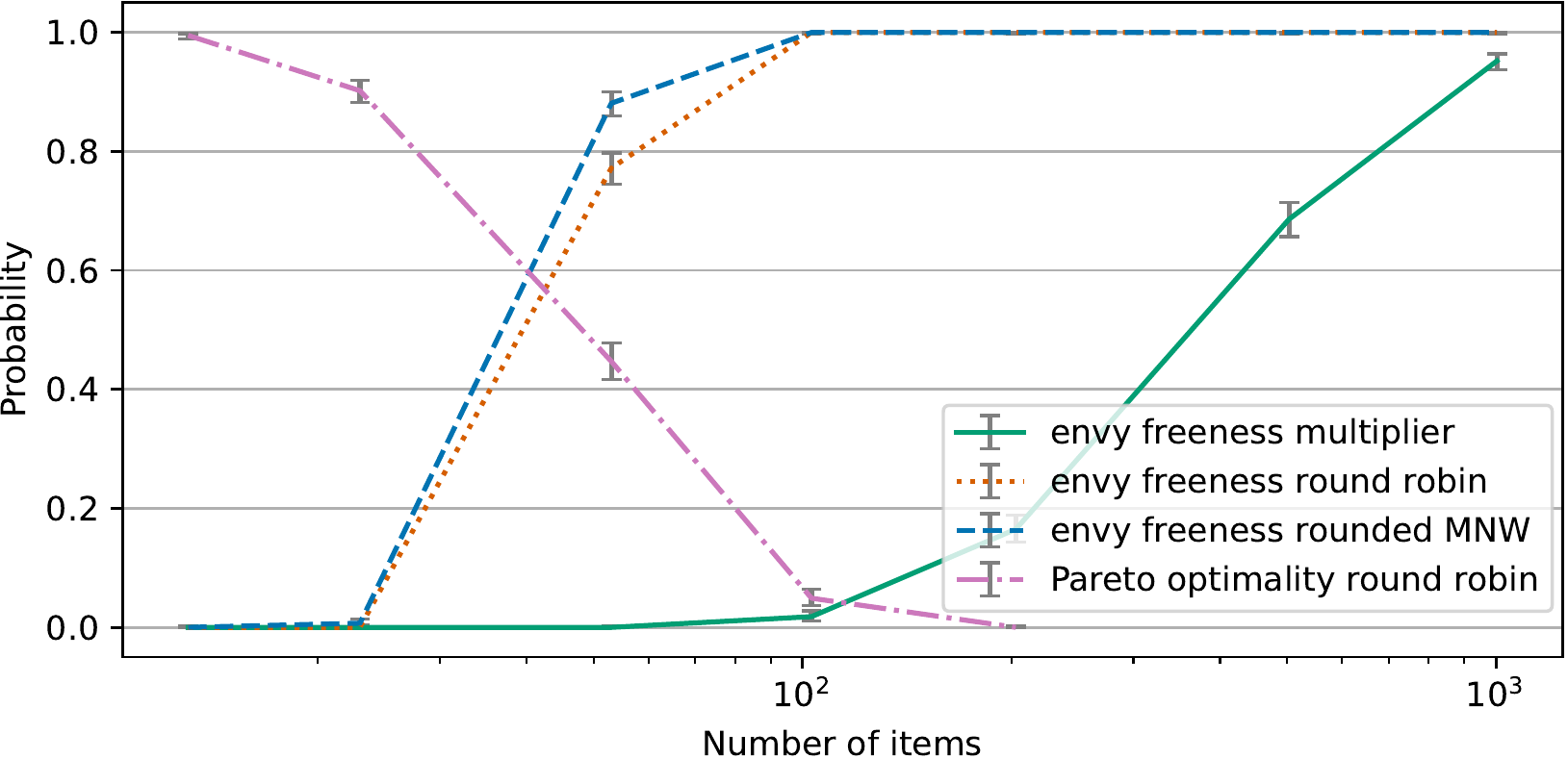}
    \caption{Version of the experiment in \cref{fig:experiments}, but with $m \equiv 3 \pmod{10}$.}
    \label{fig:experiments_offset3}
\end{figure}
We see that this shift in $m$ causes the round robin allocations and the rounded MNW allocations to converge towards envy-freeness at a slightly slower rate, and also makes the round robin algorithm be even less likely to be Pareto-optimal.
The large trends identified in the body of the paper all persist, and we see no notable difference due to the offset when $m \geq 200$.

\subsection{Experiments with Beta Distributions}
\label{app:betaexperiments}
We also study the five utility distributions in \cref{fig:scaling}:
\[ \mathcal{D}_A = \mathit{Beta}(1/2, 1/2),\quad \mathcal{D}_B = \mathit{Beta}(1, 3), \quad \mathcal{D}_C = 
\mathit{Beta}(2, 5), \quad\mathcal{D}_D = \mathit{Beta}(2, 2),\quad \mathcal{D}_E = \mathit{Beta}(5, 1). \]
We chose them based on the illustration displayed on top of the Wikipedia page on Beta distributions\footnote{See \url{https://en.wikipedia.org/wiki/Beta_distribution}, accessed on January 11, 2022. The figure is \url{https://commons.wikimedia.org/wiki/File:Beta_distribution_pdf.svg}.} at the time of writing.
Note that the distributions are not $(p,q)$-PDF-bounded for any $p,q$; however, running our algorithm with a value of $q=5$, which holds for most agents, produced multipliers of an accuracy of $10^{-5}$ in 29 seconds.

\begin{figure}[h]
\centering
    \includegraphics[width=.5\linewidth]{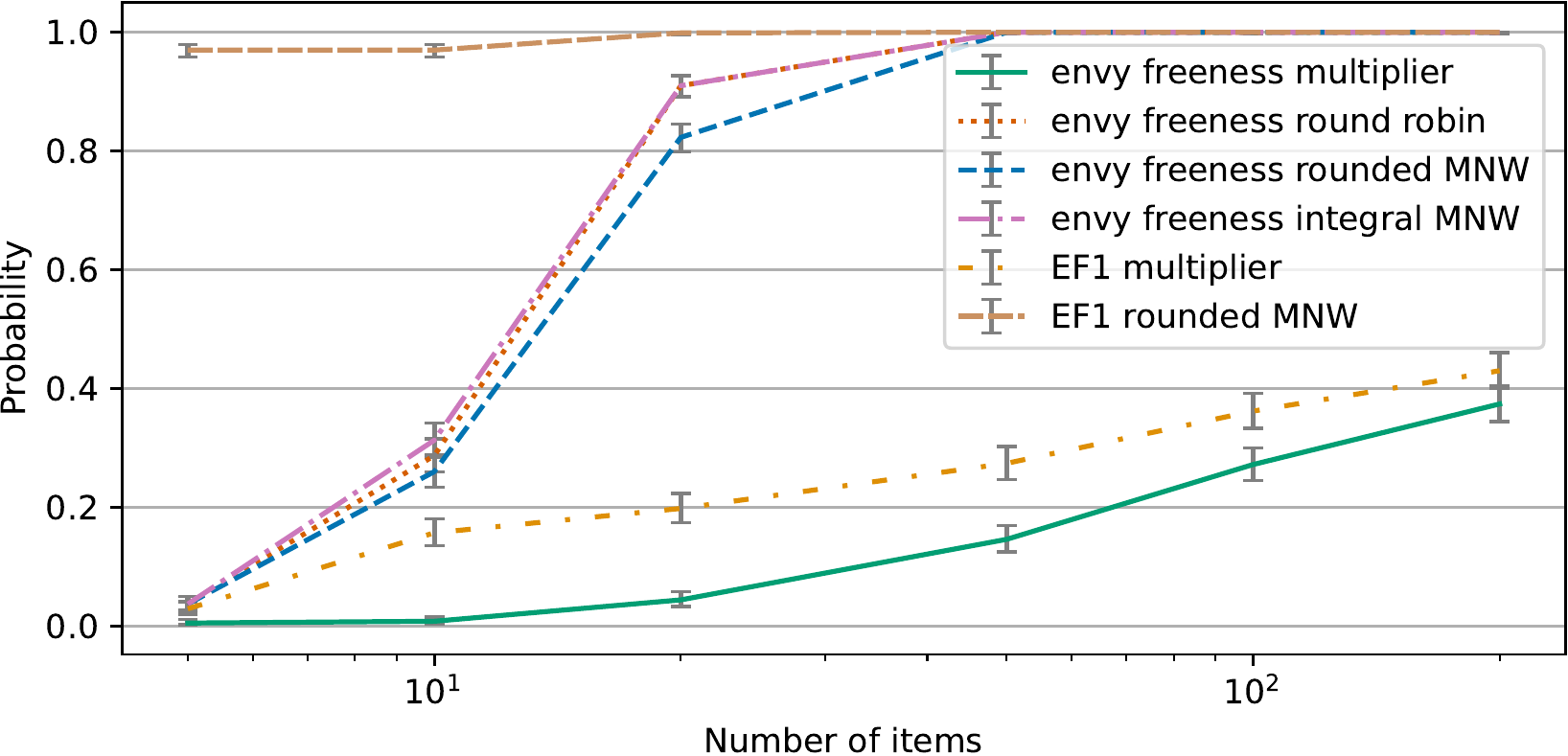}
    \caption{Version of the experiment in \cref{fig:experiments}, but with the five beta distributions and integral MNW.}
    \label{fig:experimentmnw}
\end{figure}
Since BARON ran sufficiently fast on the given five agents, this plot contains both the integral MNW allocation and the rounded fractional MNW allocation.
The figure also includes lines for when envy-freeness up to one good (EF1) holds.
This is always true for the integral MNW and round robin, the rounded MNW satisfies it almost always in our experiments, and the multiplier allocation satisfies it at a similar rate as EF.
\fi

\end{document}